\documentclass[10pt,a4paper]{article}
\usepackage[utf8]{inputenc}
\usepackage{amsmath, amsthm, amssymb}
\usepackage{enumerate}
\usepackage[svgnames]{xcolor}
\usepackage[a4paper]{geometry}
\usepackage{tikz-cd}
\newtheorem{thm}{Theorem}[section]

\newtheorem{tvrzx}[thm]{Proposition}
\newenvironment{tvrz}{\begin{tvrzx}}{\smallskip\end{tvrzx}}

\newtheorem{lemmax}[thm]{Lemma}
\newenvironment{lemma}{\begin{lemmax}}{\smallskip\end{lemmax}}

\newtheorem{theoremx}[thm]{Theorem}
\newenvironment{theorem}{\begin{theoremx}}{\smallskip\end{theoremx}}

\theoremstyle{definition}
\newtheorem{definicex}[thm]{Definition}

\theoremstyle{remark}
\newtheorem{remx}[thm]{Remark}
\newenvironment{rem}{\begin{remx}}{\medskip\end{remx}}

\theoremstyle{definition}
\newtheorem{examplex}[thm]{Example}
\newenvironment{example}{\begin{examplex}}{\medskip\end{examplex}}

\def\R{\mathbb{R}}
\def\C{\mathcal{C}}

\def\T{\mathcal{T}}
\def\<{\langle}
\def\>{\rangle}
\def\~{\widetilde}
\def\^{\wedge}

\def\g{\mathfrak{g}}
\def\h{\mathfrak{h}}

\def\io{\mathit{i}}
\def\F{\mathcal{F}}
\def\G{\mathcal{G}}
\def\B{\mathcal{B}}

\def\A{\mathcal{A}}
\def\D{\mathcal{D}}

\def\fL{\mathbf{L}}
\def\fK{\mathbf{K}}

\def\hfK{\widehat{\mathbf{K}}}
\def\di{\mathsf{d}}
\def\gm{\mathbf{G}}
\def\RS{\mathcal{R}}
\def\fPsi{\mathbf{\Psi}}

\def\cD{\nabla}

\def\hcD{\widehat{\nabla}}
\def\ocD{\overline{\nabla}}

\def\cDM{\nabla^{M}}
\def\cDL{\nabla^{LC}}
\def\cDN{\nabla^{0}}

\def\cif{C^{\infty}(M)}

\newcommand{\bm}[4]{\begin{pmatrix} #1 & #2 \\ #3 & #4 \end{pmatrix}}

\newcommand{\vf}[1]{ \mathfrak{X}^{#1}(M)}
\newcommand{\df}[1]{ \Omega^{#1}(M)}
\newcommand{\vfP}[1]{ \mathfrak{X}^{#1}(P)}
\newcommand{\dfP}[1]{ \Omega^{#1}(P)}

\newcommand{\Li}[1]{ \mathcal{L}_{#1}}

\DeclareMathOperator{\diag}{diag}
\DeclareMathOperator{\BlockDiag}{BlockDiag}

\DeclareMathOperator{\End}{End}
\DeclareMathOperator{\Hom}{Hom}
\DeclareMathOperator{\Aut}{Aut}
\DeclareMathOperator{\Ric}{Ric}
\DeclareMathOperator{\GRic}{GRic}
\DeclareMathOperator{\Div}{div}
\DeclareMathOperator{\hRic}{\widehat{R}ic}

\begin{document}
\begin{flushright}
\today
\end{flushright}
\vspace{0.7cm}
\begin{center}

\baselineskip=13pt {\Large \bf{Heterotic Reduction of Courant Algebroid Connections and Einstein-Hilbert Actions}\\}
 \vskip0.5cm
 {\it dedicated to Satoshi Watamura on the occasion of his 60th birthday}  
 \vskip0.7cm
 {\large{ Branislav Jurčo$^{1}$, Jan Vysoký$^{2,3}$}}\\
 \vskip0.6cm
$^{1}$\textit{Mathematical Institute, Faculty of Mathematics and Physics,
Charles University\\ Prague 186 75, Czech Republic, jurco@karlin.mff.cuni.cz}\\
\vskip0.3cm

$^{2}$\textit{Institute of Mathematics of the Czech Academy of Sciences \\ Žitná 25, Prague 115 67, Czech Republic, vysoky@math.cas.cz}\\
\vskip0.3cm

$^{3}$\textit{Mathematical Sciences Institute, Australian National University \\
Acton ACT 2601, Australia, jan.vysoky@anu.edu.au}\\
\vskip0.5cm
\end{center}

\begin{abstract}We discuss Levi-Civita connections on Courant algebroids.
We define an appropriate generalization of the curvature tensor and compute the corresponding scalar curvatures in the exact and heterotic case, leading to generalized (bosonic) Einstein-Hilbert type of actions known from supergravity.
In particular, we carefully analyze the process of the reduction for the generalized metric, connection, curvature tensor and the scalar curvature.

\end{abstract}


\section{Introduction}
This paper contains a thorough discussion of Riemaniann geometry on exact and heterotic Courant algebroids, i.e., within the framework of generalized geometry as introduced by Hitchin \cite{Hitchin:2004ut} and further developed in \cite{Gualtieri:2003dx, 2005math......8618H,2006CMaPh.265..131H}. The discussion here is in some aspects analogous to the Kaluza-Klein (KK) theory; See \cite{Bailin:1987jd} for a nice review of KK . In the KK theory one starts with a metric on a principal $G$-bundle $P$, with a (compact) Lie group $G$. A $G$-invariant metric on $P$ determines  (and is determined by) a metric on the base manifold $M$, a principal connection on $P$ and a $G$-invariant metric on each fibre $G_x$, smoothly depending on the base point $x$. Let us recall, to $P$ there is the associated Atiyah algebroid $L$ and the connection $A$ corresponds to a choice of splitting of the corresponding Atiyah sequence. One can use the Levi-Civita connection on $P$ and compute the Ricci scalar. The resulting Einstein-Hilbert type of action contains among others the ordinary Einstein-Hilbert action with the pure Yang-Mills theory. Also, let us recall that the KK-reduction naturally incorporates the dilaton.

Here, we modify this in the spirit of the generalized geometry. We can start with the generalized cotangent bundle $TP\oplus T^{\ast}P$ equipped with the structure of an exact Courant algebroid. If the principal action is the so call trivially extended one (and the first Pontryagin class of $P$ vanishes), the exact Courant algebroid structure on $P$ can be reduced to a Courant algebroid structure. In case of a Lie group  $G$, whose Lie algebra ${\g}$ is equipped with an ad-invariant non-degenerate bilinear form $\<\cdot ,\cdot \>_{\g}$ (e.g. compact, semisimple), the resulting Courant algebroid $E'$ is not an exact one, it is a so called heterotic Courant algebroid, its underlying vector bundle is the Whitney sum of the Atiyah algebroid $L$ and the cotangent bundle $T^{\ast}M$. Vice versa, each such a heterotic Courant algebroid comes as a reduction from an exact Courant algebroid on $P$ \cite{Baraglia:2013wua}. Similar statements can be made with respect to the respective generalized metrics. In this paper we thoroughly investigate the reduction on the level of Levi-Civita connections and the corresponding generalized Einstein-Hilbert actions. Roughly speaking, starting with an exact Courant algebroid (with the Dorfman bracket twisted by a closed 3-form $H$),  equipped with a generalized metric $(g,B)$ we arrive (ignoring the dilaton) at the generalized scalar curvature
\[ \RS = \RS(g) - \frac{1}{12} H'_{klm} H'^{klm} \]
with $H'= H+dB$. Similarly, starting with a heterotic Courant algebroid, we arrive (again in the simplest case and ignoring the dilaton and the cosmological constant) at 
\[\RS = \RS(g) - \frac{1}{12} H'_{klm} H'^{klm} + \frac{1}{4} \<F'_{kl}, F'^{kl}\>_{\g},\]
 where $F'$ is the curvature of a connection which is the sum of the starting principal connection $A$ on $P$ and an adjoint bundle valued one form $A'$ on $M$ entering the parametrization of a generalized metric $(g, B, A')$ on a heterotic Courant algebroid, and $H' = H + dB + \dots$. In this paper, among other things, we describe in detail, how the two above actions can be related by the reduction with respect to the trivially extend action of $G$. This relation will appear to be less straightforward as it might seem at the first glance. Let us note that the constants $-1/12$ and $1/4$ are related to the choice of normalizations of the fields ($H$, $B$, $A$, $A'$) as these appear naturally from the generalized geometry of Courant algebroids. E.g., $B$ is exactly the one entering the sum $g+B$.

The relevance of the heterotic Courant algebroids is due to the condition of the triviality of the first Pontryagin class. As noted, e.g., in \cite{2013arXiv1304.4294G,  Baraglia:2013wua}, it is exactly the Green-Schwarz anomaly cancellation condition when the principal bundle $P$ is a fibre product of a Yang-Mills bundle and the (oriented) frame bundle on $M$. Hence, the structure of a heterotic Courant algebroids can be used to naturally incorporate the corresponding $\alpha'$ correction. Related to this, recently, the heterotic effective actions, Green-Schwarz mechanism and the related $\alpha'$ correction have been extensively examined within the double field theory \cite{Bedoya:2014pma, Hohm:2014xsa, Hohm:2015mka, Hohm:2014eba, Marques:2015vua}\footnote{For a general review of double field theory, including discussion of effective action see \cite{Hohm:2013bwa, Aldazabal:2013sca}}. It would be interesting to compare the two approaches. Note also the generalized geometry approach to $\alpha'$ corrections published in \cite{Coimbra:2014qaa}, and generalized connections with applications in DFT in \cite{Jeon:2010rw, Jeon:2011cn}. 

Comparing to other closely related work, let us note the following.
Our definition of the torsion operator is suitable for any local (pre-)Leibniz algebroid \cite{Jurco:2015xra}. In the particular case of a Courant algebroid it can be related to the one defined by Gualtieri \cite{2007arXiv0710.2719G} or equivalently in \cite{alekseevxu}. Similarly, our notion of the curvature operator is suitable for any local (pre-)Leibniz algebroid \cite{Jurco:2015xra}. 

Efforts to construct a well-defined Riemann-like tensors encoding the low energy actions date back to Siegel \cite{Siegel:1993bj, Siegel:1993th}. In the framework of double field theory, generalized Riemann tensors were studied extensively in \cite{Hohm:2010xe} and in \cite{Hohm:2011ex} for heterotic case. In terms of geometry used in the double field theory, the generalized Riemann tensor is defined in \cite{Hohm:2012mf}. Unfortunately, while writing this paper, we were unaware of this work. Its basic idea is very similar to the approach we have taken. In this version of the paper, we have added a new section \ref{sec_dftRiemann} comparing the two approaches. We appreciate that the definition of a generalized Riemann tensor in \cite{Hohm:2012mf} has nicer symmetries (algebraic Bianchi identity) and geometrical properties. Also, it applies to a general Courant algebroid.

Recall that an important role of Courant algebroids and generalized geometry in the geometrization of supergravity was conjectured in the talk of Peter Bouwknegt \cite{talkbouwknegt}. There are many recent developments of similar ideas, see in particular the work of Coimbra, Strickland-Constable and Waldram in \cite{Coimbra:2011nw,Coimbra:2012af}. 

The paper is organized as follows. 

In Section \ref{sec_courant}, we briefly recall basic definitions related to Courant algebroids, in particular the notion of a Courant algebroid connection in \ref{sec_courantcon}. Note that this is not an ordinary vector bundle connection, as it induces the covariant derivative along a general section of the underlying Courant algebroid vector bundle $E$, not only along a vector field. Moreover, the natural compatibility with the Courant metric (pairing) is imposed. We recall the definition of the torsion operator, which has to be modified in order to reflect the non-skew-symmetry of the Dorfman bracket. The non-skew symmetry of the bracket  results in non-tensoriality of the naive torsion operator, the proper modification described here fixes this unpleasant feature,

Section \ref{sec_exact} is devoted to a thorough examination of the Courant algebroid connections in the case of exact Courant algebroids. 

In \ref{sec_genmetric}, we briefly recall all (for our purposes) useful definitions of the generalized metric, in particular the one motivated by physics and the string background $(g,B)$ encoded in the metric $g$ and the $B$-field. Let us note that unlike in some of the published work (e.g., \cite{Baraglia:2013wua} \cite{2013arXiv1304.4294G}), we keep working with an arbitrary isotropic splitting of $E$, not choosing the one by untwisting $B$ from the generalized metric. We do this to keep the explicit track of the $B$-field throughout the calculations. 

In \ref{sec_curvature}, we give a novel definition of the curvature operator, suitable for every Courant algebroid connection on an exact Courant algebroid. It is based on the more general procedure we have described in \cite{Jurco:2015xra}. Let us note that the idea of fixing the tensoriality of the curvature operator is not completely new. In, e.g.,  \cite{Vaisman:2012ke}, a vector bundle connection on $E$ is used to modify the bracket - and consequently use it in definition of the torsion and the curvature operator.

We classify all Courant algebroid connections on an exact Courant algebroid in \ref{sec_LC}. We are aware that a similar thing was done in a more abstract way in \cite{2013arXiv1304.4294G}. However, we take the more pedestrian approach and describe the result in terms of two ordinary tensors on the space-time manifold $M$. 
We define and calculate two different scalar curvatures of the most general Levi-Civita connections in \ref{sec_scalars}.

The rest of Section \ref{sec_exact} is dedicated to some immediate applications of the previous results. In \ref{sec_RicvsGric}, we compare our definition of the Ricci tensor to the generalized Ricci curvature defined in \cite{2013arXiv1304.4294G}. In \ref{sec_dilaton}, we show how the scalar field $\phi$ can be added due to a special choice of the tensor $K$ appearing in our classification, resulting in the scalar curvature containing the terms with dilaton. We compare it to a similar consideration in \cite{2013arXiv1304.4294G}. Finally, one can relate the Levi-Civita connection corresponding to the generalized metric parametrized by fields $(g,B)$ to the one corresponding to the so called background-independent gauge \cite{Seiberg:1999vs}. The tensorial character of the curvature operator allows for a direct relation of the respective scalar curvatures, leading to an elegant geometrical explanation of the correspondence of the two various effective theories (discussed, e.g., in \cite{Blumenhagen:2012nt, Blumenhagen:2013aia}. This is shown in detail in \ref{sec_BIG}. Note that this effective theory is closely related to the one discussed thoroughly in \cite{Asakawa:2015jza}. 

Section \ref{sec_heterotic} essentially generalizes the preceding section to a more general class of Courant algebroids, which we, in accordance to \cite{Baraglia:2013wua}, call heterotic Courant algebroids. 

In \ref{sec_heteroticCourant}, we recall how these heterotic Courant algebroids can be obtained by a reduction from exact Courant algebroids over principal bundles. We follow \cite{Baraglia:2013wua} and \cite{Bursztyn2007726}, \cite{2015LMaPh.tmp...53S}. In \ref{sec_gmred}, we first define a generalized metric on the heterotic Courant algebroid as an involution defining a positive definite fiber-wise metric, and show that for compact Lie groups, such generalized metrics are one-to-one with those obtained by reduction from the exact Courant algebroids in the most natural and straightforward way. The result coincides with the one in \cite{Baraglia:2013wua}. We provide the calculation for an arbitrary splitting, explicitly tracking the corresponding conditions imposed on the $B$-field. Decomposition of a generalized metric gives us the correct twisting map for the heterotic Courant algebroid bracket, an analogue of the $B$-transform in the exact theory. We calculate the resulting twisted bracket in \ref{sec_hettwist}.

In \ref{sec_hetcon}, we provide a complete classification of Levi-Civita Courant algebroid connections for the heterotic case. Again, this should be compared to \cite{2013arXiv1304.4294G}.  However, as in the exact case, we take the more pedestrian approach and describe the result in terms of ordinary tensors on the space-time manifold $M$. The result has more freedom than in the exact case, notably there is no canonical choice of a what one could call a "minimal connection". In \ref{sec_hetcurv}, it turns about that the definition of a curvature operator suitable for heterotic Courant algebroids requires a modification of the correcting term $\fK$, containing a peculiar choice of the factor $\frac{1}{2}$ in order for the resulting scalar curvatures to contain the field $\C$ of the generalized metric only in terms of the twisted bracket. This is similar to the exact case, where $B$ appears only as $dB$ in the result. We then give an example of the Levi-Civita connection and we calculate its scalar curvatures. Similarly to the exact case, we also discuss the addition of the dilaton.

In addition to the previous section, results of the Section \ref{sec_reduction} could be considered as the main achievements of this paper. As the generalized metrics of the heterotic theory are all obtained by the reduction of the relevant generalized metrics on the exact Courant algebroid over the corresponding principal bundle, one can expect that a similar procedure can be used to reduce the Levi-Civita connections.  This is indeed true, as we demonstrate in \ref{sec_conrel}. However, the correspondence of the tensors parametrizing the respective connections is not straightforward at all. One can see it from the provided example. Nevertheless, we were able to find it explicitly for one of the "minimal" connections on the heterotic Courant algebroid. 

More importantly, the correspondence of the connections provides a relation of their respective curvature operators, and consequently also of the two scalar curvatures. This calculation is the main subject of \ref{sec_conreduction}, resulting in the Theorem \ref{thm_scalarreduction}. We apply this to the example of the two corresponding connections of the previous subsection, ending up with an explicit and interesting relation (\ref{eq_connrelfin1}) of the two scalar curvatures. 

Finally, we discuss the dilaton within the reduction framework, this might be surprisingly a bit less straightforward than one would expect.  

As we also remark, starting from pre-Courant algebroids, everything can  immediately generalized be to more general principal bundles with non-vanishing first Pontryagin class. Also, one can include Lorentzian manifolds into all considerations present in the paper. 

In \ref{sec_dftRiemann} we compare in detail the curvature tensor defined here and the one defined in \cite{Hohm:2012mf} within the double field theory.

Although we have focused on principal bundles with semisimple compact structure Lie group, one can easily work out very similar results for some other examples, e.g. torus bundles. On the other hand, one can consider the class of isotropic trivially extended actions, where exact Courant algebroids reduce to exact Courant algebroids.   In this case, the reduction procedure requires one to take the quotient with respect to an isotropic subspace of the Courant metric, which poses a serious problem for the reduction of the positive definite generalized metric. One can thus rightly expect similar problems with with the reduction of corresponding Levi-Civita connections. Possible techniques required to generalize \ref{sec_conreduction} would therefore have to be more involved, and we keep it for future discussions.  
\section{Courant algebroids} \label{sec_courant}
Let us briefly recall some basic definitions. Let $E$ be a vector bundle over a manifold $M$. Let $\rho \in \Hom(E,TM)$ be a smooth vector bundle morphism, called the \emph{anchor}. Let $[\cdot,\cdot]_{E}$ be an $\R$-bilinear bracket on $\Gamma(E)$, the module of smooth sections of $E$. We say that $(E,\rho,[\cdot,\cdot]_{E})$ is a \emph{Leibniz algebroid}\footnote{Also known as Loday algebroid in mathematical literature.}, if the \emph{Leibniz rule}
\begin{equation} \label{def_courant1}
[e,fe']_{E} = f[e,e']_{E} + (\rho(e).f) e', 
\end{equation}
holds for all $e,e' \in \Gamma(E)$ and $f \in \cif$, and the bracket $[\cdot,\cdot]_{E}$ defines a Leibniz algebra on $\Gamma(E)$, i.e.,
\begin{equation} \label{def_courant2}
[e,[e',e'']_{E}]_{E} = [[e,e']_{E},e'']_{E} + [e', [e,e'']_{E}]_{E},
\end{equation}
holds for all $e,e',e'' \in \Gamma(E)$. This condition is called the \emph{Leibniz identity}. It follows from these two axioms that the  anchor is a bracket homomorphism, that is
\begin{equation} \label{eq_anchorhom}
\rho([e,e']_{E}) = [\rho(e),\rho(e')], 
\end{equation}
holds for all $e,e' \in \Gamma(E)$. 

Now, assume that $\<\cdot,\cdot\>_{E}$ is a (not necessarily a positive definite) fiber-wise metric on $E$. One says that $(E,\rho,\<\cdot,\cdot\>_{E},[\cdot,\cdot]_{E})$ is a \emph{Courant algebroid}, if $(E,\rho,[\cdot,\cdot]_{E})$ is a Leibniz algebroid and the following two relations 
\begin{equation} \label{def_courant3}
\< [e,e], e' \>_{E} = \frac{1}{2}\rho(e').\<e,e\>_{E}, 
\end{equation}
\begin{equation} \label{def_courant4}
\rho(e).\<e',e''\>_{E} = \<[e,e']_{E}, e''\>_{E} + \<e', [e,e'']_{E}\>_{E}, 
\end{equation}
hold for all $e,e',e'' \in \Gamma(E)$. 

In order to understand the first of the above conditions, let $g_{E} \in \Hom(E,E^{\ast})$ be the vector bundle isomorphism induced by $\<\cdot,\cdot\>_{E}$. To start, define the map $\di: \cif \rightarrow E^{\ast}$ as $\di{f} := \rho^{T}(df)$, where $\rho^{T} \in \Hom(T^{\ast}M,E)$ is the natural transpose of the anchor $\rho$. Then, one can use $g_{E}$ to induce the map $\D: \cif \rightarrow E$ defined as $\D{f} := g_{E}^{-1}(\di{f})$. Note that this map satisfies the usual Leibniz rule in the form 
\begin{equation} \D{(fg)} = (\D{f}) g + f (\D{g}), \end{equation}
for all $f,g \in \cif$. Using this map, one can rewrite (\ref{def_courant3}) as 
\begin{equation}
[e,e']_{E} = -[e',e]_{E} + \D\<e,e'\>_{E},
\end{equation}
for all $e,e' \in \Gamma(E)$. We see that the combination of $\rho$ and $g_{E}$ is used to control the symmetric part of the bracket. The axiom (\ref{def_courant4}) is a compatibility of the pairing $\<\cdot,\cdot\>_{E}$ with the bracket $[\cdot,\cdot]_{E}$. Courant algebroids should thus be considered as a natural algebroid generalization of quadratic Lie algebras. In particular, they can be recovered as Courant algebroids over $M = \{m\}$. 

One says that $(E,\rho,\<\cdot,\cdot\>_{E},[\cdot,\cdot]_{E})$ is an \emph{exact Courant algebroid}, if the sequence 
\begin{equation} \label{eq_courantSES}
\begin{tikzcd}
0 \arrow[r] & T^{\ast}M \arrow[r, "\rho^{\ast}"] & E \arrow[r,"\rho"] & TM \arrow[r] & 0, 
\end{tikzcd}
\end{equation}
is exact. Here $\rho^{\ast} = g_{E}^{-1} \circ \rho^{T}$, that is $\rho^{\ast}(df) = \D{f}$. Note that $\rho \circ \rho^{\ast} = 0$ holds for a general Courant algebroid. 

It was proved by Ševera in \cite{severaletters} that exact Courant algebroids over $M$ are classified by $H^{3}_{dR}(M)$. 
In particular, every exact Courant algebroid over $M$ has (up to an isomorphism) the following form: $E := TM \oplus T^{\ast}M$, $\rho$ is the canonical projection onto vector fields, $\<\cdot,\cdot\>_{E}$ is the canonical pairing of vector fields and $1$-forms $\<X+\xi,Y+\eta\>_{E}:= \eta(X)+ \xi(Y)$ and $[\cdot,\cdot]_{E}$ is the $H$-twisted Dorfman bracket 
\begin{equation} \label{def_Hdorfman}
[X+\xi,Y+\eta]_{D}^{H} := [X,Y] + \Li{X}\eta  - \io_{Y}d\xi - H(X,Y,\cdot),
\end{equation}
for all $X+\xi, Y+\eta \in \Gamma(E)$. The form $H \in \Omega^{3}(M)$ has to be closed. Let $H' \in \Omega^{3}_{closed}(M)$. Then $[\cdot,\cdot]_{D}^{H}$ and $[\cdot,\cdot]_{D}^{H'}$ are isomorphic iff $[H]_{dR} = [H']_{dR}$. 

\subsection{Courant algebroid connections} \label{sec_courantcon}
Let $(E,\rho,\<\cdot,\cdot\>_{E},[\cdot,\cdot]_{E})$ be a Courant algebroid. We follow the definitions in \cite{alekseevxu} and \cite{2007arXiv0710.2719G}. Let $\cD: \Gamma(E) \times \Gamma(E) \rightarrow \Gamma(E)$ be an $\R$-bilinear map. We say that $\cD$ is a \emph{Courant algebroid connection} on $E$ if 
\begin{equation}
\cD_{fe}e') = f \cD_{e}e', \; \cD_{e}fe' = f \cD_{e}e + (\rho(e).f)e'
\end{equation}
holds for all $e,e' \in \Gamma(E)$ and $f \in \cif$ together with the metric compatibility condition
\begin{equation} \label{def_courantconcomp}
\rho(e).\<e',e''\>_{E} = \< \cD_{e}e', e'' \>_{E} + \<e', \cD_{e}e''\>_{E},
\end{equation}
for all $e,e',e'' \in \Gamma(E)$. We have used the conventional notation $\cD_{e}e' := \cD(e,e')$. We say that $\cD$ is an \emph{induced} Courant algebroid connection, if there is an ordinary vector bundle connection $\cD'$ on $E$, such that $\cD_{e} = \cD'_{\rho(e)}$. Using the standard procedure, we can extend the covariant derivative to the whole tensor algebra $\T(E)$ of the vector bundle $E$. 
Let $g_{E}$ be a fiber-wise metric on $E$. We say that $\cD$ is \emph{metric compatible} with $g_{E}$, if $\cD_{e}g_{E} = 0$, or equivalently
\begin{equation} 
\rho(e).g_{E}(e',e'') = g_{E}(\cD_{e}e',e'') + g_{E}(e',\cD_{e}e''),
\end{equation}
for all $e,e',e'' \in \Gamma(E)$. 

An introduction of torsion operator is less straightforward. We will a definition equivalent to those in \cite{alekseevxu,2007arXiv0710.2719G}. By a \emph{torsion operator} we mean a $\cif$-bilinear map $T: \Gamma(E) \times \Gamma(E) \rightarrow \Gamma(E)$ defined as
\begin{equation} \label{def_torsion}
T(e,e') := \cD_{e}e' - \cD_{e'}e - [e,e']_{E} + \< \cD_{e_{\lambda}}e, e'\>_{E} \cdot g_{E}^{-1}(e^{\lambda}), 
\end{equation}
for all $e,e' \in \Gamma(E)$. Here $(e_{\lambda})_{\lambda=1}^{k}$ is an arbitrary local frame on $E$, and $(e^{\lambda})_{\lambda=1}^{k}$ is the corresponding dual frame on $E^{\ast}$. It is straightforward to check the $\cif$-linearity in both inputs. Moreover, $T$ is in skew-symmetric in $(e,e')$. In fact, there holds even stronger statement. Note that the $3$-form $T_{G}$ in the following lemma is the original definition of torsion according to Gualtieri in \cite{2007arXiv0710.2719G}. 
\begin{lemma} \label{lem_torsgualt}
let $T_{G} \in \T_{3}^{0}(E)$ defined as $T_{G}(e,e',e'') := \< T(e,e'), e''\>_{E}$. Then $T_{G}$ is completely skew-symmetric, that is $T_{G} \in \Omega^{3}(E)$. 
\end{lemma}

\section{Exact theory} \label{sec_exact}
\subsection{Generalized metric} \label{sec_genmetric}
In this section, we assume that $(E,\rho,\<\cdot,\cdot\>_{E},[\cdot,\cdot]_{E})$ is an exact Courant algebroid with the twisted Dorfman bracket (\ref{def_Hdorfman}). There are several equivalent ways to define a generalized metric in this case. Let us start with the following one. We say than $\tau \in \End(E)$ is a \emph{generalized metric} on $E$, if $\tau^{2} = 1$, and the formula
\begin{equation}
\gm_{\tau}(e,e'):= \< \tau(e), e'\>_{E},
\end{equation}
defines a positive definite fiber-wise metric on $E$. This also implies that $\tau$ is symmetric and orthogonal with respect to $\<\cdot,\cdot\>_{E}$. For exact $E$, the signature of $\<\cdot,\cdot\>_{E}$ is $(n,n)$. Denote the group of orthogonal automorphisms of $E$ as $O(E)$, and the corresponding Lie algebra of skew-symmetric maps as $o(E)$. In this case, the $\pm 1$ eigenbundles $V_{\pm}$ of $\tau$ define rank $n$ positive and negative definite subbundles of $E$, such that 
\begin{equation}
E = V_{+} \oplus V_{-}.
\end{equation}
Moreover, one has $V_{-} = V_{+}^{\perp}$. A choice of a maximal rank positive definite subbundle $V_{+}$ is equivalent to a choice of a generalized metric. Given $V_{+}$, define $V_{-} := V_{+}^{\perp}$, and $\tau|_{V_{\pm}} = \pm 1$. Next, let $L \subseteq E$ be a maximally isotropic subbundle. One can identify the orthogonal complement $L^{\perp}$ with the dual bundle $L^{\ast}$, hence $E = L \oplus L^{\ast}$. Because $L$ and $L^{\ast}$ are isotropic, we have $L \cap V_{+} = L^{\ast} \cap V_{+} = 0$. This implies that $V_{+}$ is a graph of a unique vector bundle isomorphism $A \in \Hom(L,L^{\ast})$. In other words, $V_{+}$ has the form
\begin{equation}
V_{+} = \{ \psi + A(\psi) \; | \; \psi \in L \} \subseteq L \oplus L^{\ast}. 
\end{equation}
We can find a unique decomposition $A = g + B$, where $g \in \Gamma(S^{2}L^{\ast})$, and $B \in \Omega^{2}(L)$. It follows that $g$ has to be positive definite fiber-wise metric on $L$. Using the similar arguments and perpendicularity of two subbundles, we get that 
\begin{equation}
V_{-} = \{ \psi + (-g + B)(\psi) \; | \; \psi \in L \} \subseteq L \oplus L^{\ast}. 
\end{equation}
Thus, given a maximally isotropic subbundle $L$, to every generalized metric there exists a unique pair $(g,B)$, where $g$ is a positive definite fiber-wise metric on $L$, and $B \in \Omega^{2}(L)$. Conversely, having the data $(g,B)$, set $V_{+}$ to be as above in order to define a generalized metric. The corresponding fiber-wise metric $\gm_{\tau}$ on $E$ can be written in the formal block form with respect to the splitting $E = L \oplus L^{\ast}$ as 
\begin{equation}
\gm_{\tau} = \bm{g - Bg^{-1}B}{Bg^{-1}}{-g^{-1}B}{g^{-1}}.
\end{equation}
We will drop the subscript $\tau$ in the following. Note that $\gm$ can be decomposed as
\begin{equation}
\gm = \bm{1}{B}{0}{1} \bm{g}{0}{0}{g^{-1}} \bm{1}{0}{-B}{1}.
\end{equation}
For simplicity, we always choose the maximally isotropic subbundle $L = TM$, but keep in mind that everything works in the same way for an arbitrary $L$. 

Let $B \in \df{2}$. Let $\F_{B} \in o(E)$ have the form $\F_{B}(X+\xi) = B(X)$. Taking its exponential, we get the map $e^{B} \in O(E)$, which has the block form
\begin{equation}
e^{B} = \bm{1}{0}{B}{1}.
\end{equation}
The map $e^{B}$ is usually called the \emph{$B$-transform}. We see that $\gm$ can be written as $\gm = (e^{-B})^{T} \G_{E} e^{-B}$, where $\G_{E}$ is the block diagonal fiber-wise metric $\G_{E} = \BlockDiag(g,g^{-1})$. A significant feature is the behavior of the twisted Dorfman bracket under $B$-transform. We have
\begin{equation} \label{eq_DorftransfeB}
e^{B}([e,e']_{D}^{H+dB}) = [e^{B}(e), e^{B}(e')]_{D}^{H}.
\end{equation}
We see that $e^{B}$ is precisely the isomorphism of two twisted Dorfman brackets corresponding to $3$-forms $H$ and $H'$ in the same cohomology class. 

\begin{rem}
We use the following convention for maps induced by $2$-forms and $2$-vectors. For example, we set $B(X) = B(\cdot,X)$, where on the left-hand side $B \in \Hom(TM,T^{\ast}M)$, and on the right-hand side $B \in \df{2}$. Note that the matrix of the map $B$ in any basis coincides with components of the $2$-form in the same basis, that is $\<e_{k},B(e_{l})\> = B(e_{k},e_{l})$.
\end{rem}
\subsection{Definition of the curvature operator} \label{sec_curvature}
Assume that $E$ is an exact Courant algebroid, $E = TM \oplus T^{\ast}M$. We can now proceed to the definition of the curvature operator. This is an example of a more general procedure described in \cite{Jurco:2015xra}. The main idea is to correct the naive curvature operator formula to obtain a tensor with good properties. Define
\begin{equation} \label{def_Riemann}
R(e,e')e'' := \cD_{e}\cD_{e'}e'' - \cD_{e'}\cD_{e}e'' - \cD_{[e,e']_{E}}e'' + \cD_{\fK(e,e')}e'',
\end{equation}
where $\fK$ is defined as 
\begin{equation} \label{def_fK}
\fK(e,e'):= \< \cD_{e_{\lambda}}e, e'\>_{E} \cdot pr_{2}( g_{E}^{-1}(e^{\lambda})).
\end{equation}
Observe that $\fK$ resembles the last term of (\ref{def_torsion}), except for the projection onto $T^{\ast}M$. This modification is necessary in order to establish the tensorial property of $R$. 
\begin{lemma} \label{lem_Riemann}
The operator $R$ defined by (\ref{def_Riemann}) and (\ref{def_fK}) is $\cif$-linear\footnote{Here, following the suggestion of an anonymous referee, we clarify this claim. In particular, we refer to two versions of the paper \cite{Vaisman:2012px}. In the first version of \cite{Vaisman:2012px}, an incorrect definition similar to our (\ref{def_Riemann}) is given. The key for our version to work is the projection $pr_{2}$ in the definition (\ref{def_fK}) of the map $\fK$. Because of that, it satisfies $\fK(e,fe') = f \fK(e,e')$, $\fK(fe,e') = \<e,e'\>_{E} \cdot \D{f}$ and $\rho( \fK(e,e')) = 0$, for all $e,e' \in \Gamma(E)$ and $f \in \cif$. It is precisely the combination of these three properties which can easily be used to prove that $R$ is indeed $\cif$-linear in all inputs. For a more general discussion relevant to Leibniz alegbroids, we refer to our previous paper \cite{Jurco:2015xra}.} in $e,e',e''$. In fact, it is skew-symmetric in $(e,e')$: $R(e,e')e'' = - R(e',e)e''$. Moreover, $R$ satisfies
\begin{equation}
\<R(e,e')f,f'\>_{E} + \<R(e,e')f',f\>_{E} = 0.
\end{equation}
\end{lemma}
\begin{proof}
All claims follow by a direct verification using the axioms of Courant algebroids and metric compatibility of $\cD$ and $\<\cdot,\cdot\>_{E}$. 
\end{proof}

Having the curvature operator, or equivalently the generalized Riemann tensor, we can define the generalized Ricci tensor in a usual way as
\begin{equation} \label{def_Ricci}
\Ric(e,e') := \< e^{\lambda}, R(e_{\lambda},e')e \>. 
\end{equation}
At this point, one should compare the definition (\ref{def_Riemann}) with the one introduced in \cite{2007arXiv0710.2719G}, and subsequently used in \cite{2013arXiv1304.4294G}. In these papers, instead of fixing the non-tensoriality of $R$ by adding the term with $\fK$, they note that the naive $R(e,e')$ is tensorial in $e$, if and only if $\<e,e'\>_{E} = 0$. In particular, one can always restrict onto Dirac structure in $E$. However, our $R$ \emph{does not} in general restrict to the naive $R$ for a pair of orthogonal sections. In \cite{2013arXiv1304.4294G}, they choose $e \in V_{+}$ and $e' \in V_{-}$, where $V_{\pm} \subseteq E$ are the subbundles induced by a generalized metric. For connections compatible with $\gm$, those subbundles are invariant, $\cD_{e}(V_{\pm}) \subseteq V_{\pm}$. It follows that in this case $\fK(e,e') = 0$, and our $R$ reduces to the one of \cite{2007arXiv0710.2719G}. Moreover, our tensor $\Ric \in \Gamma(E^{\ast} \otimes E^{\ast})$ does restrict to their $\GRic \in  \Gamma(V^{\ast}_{-} \otimes V^{\ast}_{+})$. However, our curvatures $\RS$ and $\RS_{E}$ can be defined in a more straightforward way using directly the traces of the Ricci tensor. Interestingly, the resulting curvature $\RS$ coincides with the generalized scalar curvature GS, defined using the spinor bundle and the Dirac operator in \cite{2007arXiv0710.2719G}, \cite{2013arXiv1304.4294G}.

Another highly relevant approach of \cite{Hohm:2012mf} will be discussed in detail in \ref{sec_dftRiemann}.
\subsection{Levi-Civita connections on $E$} \label{sec_LC}
Assume that $E = TM \oplus T^{\ast}M$ is equipped with an $H$-twisted Dorfman bracket. Suppose $\gm$ is a generalized metric on $E$, corresponding to a pair $(g,B)$, or equivalently $\gm = (e^{-B})^{T} \G_{E} e^{-B}$, where $\G_{E} = \BlockDiag(g,g^{-1})$. The main goal of this section is to describe the most general torsion-free Courant algebroid connection on $E$, compatible with the generalized metric $\gm$.  It turns out that such a connection is not unique. Despite of this, we refer to such connections as Levi-Civita ones. We start with the following observation.
\begin{lemma}
Connection $\cD$ is a Levi-Civita connection corresponding to the generalized metric $\gm$ and bracket $[\cdot,\cdot]_{D}^{H}$ if and only if the connection $\hcD$ defined as
\begin{equation} \label{eq_cdandhcd}
\cD_{e}e': = e^{B} \hcD_{e^{-B}(e)} e^{-B}(e'), 
\end{equation}
is a Levi-Civita connection corresponding to $\G_{E}$ and the bracket $[\cdot,\cdot]_{D}^{H+dB}$. 
\end{lemma}
\begin{proof}
Let $\cD$ be a Levi-Civita connection corresponding to $\gm$ and $[\cdot,\cdot]_{D}^{H}$. It is easy to see that $\hcD$ is metric compatible with $\G_{E}$ and $\<\cdot,\cdot\>_{E}$. Let $\widehat{T}$ be the torsion operator corresponding to $\hcD$ and the bracket $[\cdot,\cdot]_{D}^{H+dB}$. It follows from (\ref{eq_DorftransfeB}) that 
\begin{equation}
T(e,e') = e^{B} (\widehat{T}( e^{-B}(e), e^{-B}(e'))). 
\end{equation}
This proves that $T = 0$ iff $\widehat{T} = 0$. 
\end{proof}
This lemma allows us to simplify our problem to the following one: To find all Levi-Civita connections with respect to the block diagonal generalized metric $\G_{E'}$ and the bracket $[\cdot,\cdot]_{D}^{H'}$, where the closed $3$-form $H'$ is $H' = H + dB$. It is not difficult to see that $\hcD$ is metric compatible with $\<\cdot,\cdot\>_{E}$ and $\G_{E}$, if and only it has have the block form
\begin{equation} \label{eq_hcDcompatible}
\hcD_{X} = \bm{\cDM_{X}}{g^{-1}C_{X}(\cdot,g^{-1}(\star))}{C_{X}(\cdot,\star)}{\cDM_{X}}, \; \hcD_{\xi} = \bm{A_{\xi}(\star)}{g^{-1}V_{\xi}(\cdot,g^{-1}(\star))}{V_{\xi}(\cdot,\star)}{-A_{\xi}^{T}(\star)},
\end{equation}
where $C_{X}, V_{\xi} \in \df{2}$, $\cDM$ is an ordinary linear connection on the manifold $M$ compatible with the metric $g$, and $A_{\xi} \in \End(TM)$ is skew-symmetric with respect to the metric $g$: 
\begin{align}
g( A_{\xi}(Y),Z) + g(Y, A_{\xi}(Z)) = 0.
\end{align}
All objects are $\cif$-linear in $X,\xi$ and $\star$ indicates the input. Equivalently, we can write
\begin{align}
\hcD_{X}(Z+\zeta) &= \cDM_{X}Z + g^{-1}C_{X}(\cdot,g^{-1}(\zeta)) + C_{X}(\cdot,Z) + \cDM_{X}\zeta, \\
\hcD_{\xi}(Z+\zeta) &= A_{\xi}(Z) + g^{-1}V_{\xi}(\cdot,g^{-1}(\zeta)) + V_{\xi}(\cdot,Z) - A_{\xi}^{T}(\zeta). 
\end{align}
Our goal is to determine $C,V,\cDM$ and $A$ in order to define a torsion-free connection. Plugging the expressions (\ref{eq_hcDcompatible}) into the condition $\widehat{T}(e,e') = 0$ gives a set of four independent equations: 
\begin{align}
\label{eq_tf1} \cDM_{X}Y - \cDM_{Y}X - [X,Y] + V^{k}(X,Y)\partial_{k} &= 0, \\
\label{eq_tf2} C_{X}(Z,Y) - C_{Y}(Z,X) - C_{Z}(X,Y) + H'(X,Y,Z) &= 0, \\
\label{eq_tf3} C_{X}( g^{-1}(\xi), g^{-1}(\eta)) + \< A_{\xi}(X), \eta\> - \<A_{\eta}(X),\xi\> &= 0, \\
\label{eq_tf4} V_{\xi}(g^{-1}(\zeta),g^{-1}(\eta)) - V_{\eta}(g^{-1}(\zeta), g^{-1}(\xi)) - V_{\zeta}(g^{-1}(\xi), g^{-1}(\eta)) &= 0.
\end{align}
All conditions are supposed to hold for all vector fields and $1$-forms on $M$. 
To proceed further, note that (\ref{eq_tf1}) shows that $V_{\xi}(X,Y) = \< T^{M}(X,Y), \xi \>$, where $T^{M}$ is the torsion operator of the linear connection $\cDM$. One can plug this expression into (\ref{eq_tf4}) in order to obtain the condition
\begin{equation} \label{eq_tf5}
g(T^{M}(Z,Y),X) - g(T^{M}(Z,X),Y) - g(T^{M}(X,Y),Z) = 0.
\end{equation}
Recall that the contortion tensor $K$ for metric compatible connection is defined as 
\begin{equation}
K(X,Y,Z) := \frac{1}{2} \{ g(T^{M}(X,Y),Z) + g(T^{M}(Z,X),Y) - g(T^{M}(Y,Z),X) \}.
\end{equation}
Using the condition (\ref{eq_tf5}), this gives $K(X,Y,Z) = -g(X,T^{M}(Y,Z))$. Rewriting (\ref{eq_tf5}) using the contortion tensor now gives the equation
\begin{equation} \label{eq_tf6} K(X,Y,Z) + cyclic(X,Y,Z) = 0. \end{equation}
The contortion tensor is by definition skew-symmetric in last two inputs. The condition (\ref{eq_tf6}) thus says that complete skew-symmetrization $K_{a}$ of $K$ vanishes. Now note that we can write
\begin{align}
V_{\xi}(X,Y) &= - K(g^{-1}(\xi),X,Y), \\
\cDM_{X}Y &= \cDL_{X}Y + g^{-1}K(X,Y,\cdot),
\end{align}
where $\cDL$ is the Levi-Civita connection on $M$ corresponding to the metric $g$. 	The solution to the other two equations is similar. Define the tensor $Q$ as 
\begin{equation}
Q(\xi,\eta,\zeta) := \< \eta, A_{\xi}(g^{-1}(\zeta)) \> = g^{-1}( A_{\xi}^{T}(\eta),\zeta). 
\end{equation}
Because $A_{\xi}$ is supposed to be skew-symmetric with respect to $g$, $Q(\xi,\eta,\zeta)$ has to be skew-symmetric in $(\eta,\zeta)$. Equation (\ref{eq_tf3}) then implies
\begin{equation}
C_{X}(Y,Z) = Q( g(Z),g(Y),g(X)) - Q(g(Y),g(Z),g(X)). 
\end{equation}
One can now plug this into (\ref{eq_tf2}) in order to obtain 
\begin{equation}
Q(g(X),g(Y),g(Z)) + cyclic(X,Y,Z) = - \frac{1}{2} H'(X,Y,Z). 
\end{equation}
This determines the complete skew-symmetrization of $Q$. We can thus write 
\begin{equation}
Q(\xi,\eta,\zeta) = - \frac{1}{6}H'(g^{-1}(\xi),g^{-1}(\eta),g^{-1}(\zeta)) + J(\xi,\eta,\zeta),
\end{equation}
where $J \in \vf{1} \otimes \vf{2}$ satisfies $J_{a} = 0$. One can now rewrite the remaining tensor fields using $H'$ and $J$ to obtain 
\begin{align}
A_{\xi}(Y) &= \frac{1}{6} H'(g^{-1}(\xi),Y,\cdot) - J(\xi,g(Y),\cdot), \\
C_{X}(Y,Z) &= \frac{1}{3} H'(X,Y,Z) + J(g(X),g(Y),g(Z)). 
\end{align}
We can summarize the above discussion in the form of a theorem
\begin{theorem} \label{thm_generalLC}
A Courant algebroid connection $\hcD$ is a Levi-Civita connection with respect to the generalized metric $\G_{E}$ and $H'$-twisted Dorfman bracket $[\cdot,\cdot]_{D}^{H'}$, if and only if it has the form 
\begin{align}
\label{eq_hcDexpl1} \hcD_{X} &= \bm{\cDL_{X} + g^{-1}K(X,\star,\cdot)}{-\frac{1}{3} g^{-1}H'(X,g^{-1}(\star),\cdot) - J(g(X),\star,\cdot)}{-\frac{1}{3}H'(X,\star,\cdot) - gJ(g(X),g(\star),\cdot)}{\cDL_{X} + K(X,g^{-1}(\star),\cdot)}, \\
\label{eq_hcDexpl2} \hcD_{\xi} & = \bm{ \frac{1}{6} g^{-1}H'(g^{-1}(\xi),\star,\cdot) - J(\xi,g(\star),\cdot)}{ g^{-1}K(g^{-1}(\xi),g^{-1}(\star),\cdot)}{K(g^{-1}(\xi),\star,\cdot)}{ \frac{1}{6} H'(g^{-1}(\xi),g^{-1}(\star),\cdot) - gJ(\xi,\star,\cdot)},
\end{align}
where $K \in \df{1} \otimes \df{2}$ and  $J \in \vf{} \otimes \vf{2}$ satisfy $K_{a} = J_{a} = 0$. 
\end{theorem}
We see that the space of Levi-Civita connections for a given generalized metric is an infinite-dimensional affine space over the $\cif$-module of rank $\frac{2}{3}(n^{3} - n)$. Also note that no Levi-Civita connection is in general (whenever $H' \neq 0$) induced by an ordinary connection on vector bundle $E$, that is always $\hcD_{\xi} \neq 0$, and consequently also $\cD_{\xi} \neq 0$.
\subsection{Scalar curvatures of Levi-Civita connections} \label{sec_scalars}
Let $\cD$ be the Levi-Civita Courant algebroid connection corresponding to the generalized metric $\gm$ and the bracket $[\cdot,\cdot]_{D}^{H}$. We assume that the connection $\hcD$ corresponding to $\cD$ via (\ref{eq_cdandhcd}) has the form as in Theorem \ref{thm_generalLC}. We can calculate the curvature operator (\ref{def_Riemann}) of $\cD$, and the corresponding Ricci tensor (\ref{def_Ricci}).   The ultimate goal of this section is to arrive to a pair $(\RS,\RS_{E})$ of scalar curvatures, defined as
\begin{equation} \label{def_scalars}
\RS := \Ric( \gm^{-1}(e^{\lambda}), e_{\lambda}), \; \RS_{E} = \Ric( g_{E}^{-1}(e^{\lambda}), e_{\lambda}). 
\end{equation}
To start, note that we have
\begin{equation} \label{eq_relRhatR}
R(e,e')e'' = e^{B} [ \widehat{R}(e^{-B}(e), e^{-B}(e'))( e^{-B}(e''))], 
\end{equation}
where the operator $\widehat{R}$ has the form
\begin{equation} \label{eq_hRiemann}
\widehat{R}(e,e')e'' := \hcD_{e}\hcD_{e'}e'' - \hcD_{e'}\hcD_{e}e'' - \hcD_{[e,e']_{D}^{H'}}e'' + \hcD_{\hfK(e,e')} e''.
\end{equation}
The map $\hfK$ is defined as 
\begin{equation} \label{eq_hfK}
\hfK(e,e') := \< \hcD_{e_{\lambda}}e, e'\>_{E} \cdot (e^{-B} \circ pr_{2} \circ e^{B}) (g_{E}^{-1}(e^{\lambda})). 
\end{equation}
Note that $\hfK$ and consequently also $\widehat{R}$ depend explicitly on $B$. Define the tensor $\hRic$ using $\widehat{R}$ as 
\begin{equation}
\hRic(e,e') := \< e^{\lambda}, \widehat{R}(e_{\lambda},e')e \>. 
\end{equation}
From the definition, we obtain  $\Ric(e,e') = \hRic(e^{-B}(e), e^{-B}(e'))$, and thus 
\begin{equation} \label{eq_scalrstwist}
\RS = \hRic( \G_{E}^{-1}(e^{\lambda}),e_{\lambda}), \; \RS_{E} = \hRic(g_{E}^{-1}(e^{\lambda}), e_{\lambda}). 
\end{equation}
Explicit formulas for $\widehat{R}$ and $\hRic$ can be now calculated from (\ref{eq_hcDexpl1}, \ref{eq_hcDexpl2}). We are interested merely in the scalar functions $(\RS, \RS_{E})$. We state the final result in the form of a theorem. 

\begin{theorem} \label{thm_scalarcurvatures}
Let $\cD$ be a Levi-Civita connection corresponding to the generalized metric $\gm$ and the bracket $[\cdot,\cdot]_{D}^{H}$. Let $H' = H + dB$, and let $\hcD$ be the connection (\ref{eq_cdandhcd}) having the explicit form (\ref{eq_hcDexpl1}, \ref{eq_hcDexpl2}). Let $\RS$ and $\RS_{E}$ be the scalar curvatures (\ref{def_scalars}). Then 
\begin{align}
\RS &= \RS(g) - \frac{1}{12} H'_{klm} H'^{klm} + 4 \Div( K') - 4 \rVert K' \rVert_{g}^{2} - 4 \rVert J' \rVert_{g}^{2}, \\
\RS_{E} &= - 4 \Div( J' ) + 8 \< J', K' \>,
\end{align}
where $\RS(g)$ is the Levi-Civita Ricci scalar of $g$, $J' \in \vf{}$ is defined as $J' := J(dy^{k}, g(\partial_{k}), \cdot)$, $K' \in \df{1}$ is $K' = K( g^{-1}(dy^{k}), \partial_{k}, \cdot)$, $\rVert \cdot \rVert_{g}$ are the $g$-norms of vector and covector fields, and $\Div$ is a covariant divergence induced by the Levi-Civita connection $\cDL$. 
\end{theorem}
\begin{rem}
Let us note that in the entire process, it was not necessary for $g$ to be positive definite. The calculation thus works for $g$ of any signature, as long as $g$ is non-degenerate.
\end{rem}

\subsection{Comparing $\Ric$ to $\GRic$ of  \cite{2013arXiv1304.4294G} } \label{sec_RicvsGric}
At this point, we can compare the conditions imposed by the vanishing of the generalized Ricci tensor $\GRic$ as it was discussed in \cite{2013arXiv1304.4294G} to the properties of the tensor $\Ric$ introduced here. Note that $\GRic \in \Gamma(V_{-}^{\ast} \otimes V_{+}^{\ast})$, and we get
\begin{equation}
\GRic(e_{-},e_{+}) = \Ric( e_{+}, e_{-}),
\end{equation}
where $e_{-} \in \Gamma(V_{-})$ and $e_{+} \in \Gamma(V_{+})$. In the above formula, the interchange of $e_{-}$ with $e_{+}$ follows from the opposite convention to define the Ricci tensor. They prove that $\GRic = 0$ gives the part of the equations of motion for type II supergravity. Let us now show that $\GRic = 0$ does not imply $\Ric = 0$. We will use the simplest $J = K = 0$ example. 
Define the four tensors $\Ric_{\pm \pm} \in \df{1} \otimes \df{1}$ as 
\begin{equation} \Ric_{\pm \pm}(X,Y) := \Ric( \fPsi_{\pm}(X), \fPsi_{\pm}(Y)), \end{equation}
where $\fPsi_{\pm}$ are the generalized metric induced isomorphisms $\fPsi \in \Hom(TM,V_{\pm})$ having the form
\begin{equation} \label{def_fPsipm}
\fPsi_{\pm}(X) = X + (\pm g + B)(X).
\end{equation}
As $\fPsi_{\pm}$ are isomorphisms, the condition $\GRic = 0$ is equivalent to $\Ric_{+-} = 0$. One finds
\begin{align}
\Ric_{+-}(X,Y) &= \Ric^{LC}(X,Y) - \frac{1}{4} (H' \star H')(X,Y) - \frac{1}{2} (\cD^{LC}H')(X,Y), \\
\Ric_{-+}(X,Y) &= \Ric^{LC}(X,Y) - \frac{1}{4} (H' \star H')(X,Y) + \frac{1}{2} (\cD^{LC}H')(X,Y),
\end{align}
where $(H' \star H')(X,Y) := H'(\partial_{k},\partial_{l},X) H'(e^{k},e^{l},Y)$, and $(\cD^{LC}H)(X,Y) = (\cD^{LC}_{\partial_{k}}H)(X,Y,e^{k})$, and similarly
\begin{align}
\Ric_{++}(X,Y) &= \Ric^{LC}(X,Y) - \frac{1}{12} (H' \star H')(X,Y) + \frac{1}{6} (\cD^{LC}H')(X,Y) \\
& + \frac{1}{18} H'(g^{-1}(dy^{l}),\partial_{k},Y) H'(g^{-1}B(\partial_{l}),X,g^{-1}(dy^{k})), \nonumber\\
\Ric_{--}(X,Y) &= \Ric^{LC}(X,Y) - \frac{1}{12} (H' \star H')(X,Y) - \frac{1}{6} (\cD^{LC}H')(X,Y) \\
& - \frac{1}{18} H'(g^{-1}(dy^{l}),\partial_{k},Y) H'(g^{-1}B(\partial_{l}),X,g^{-1}(dy^{k})). \nonumber
\end{align}
The condition $\GRic(X,Y) = 0$ implies $\Ric_{+-}(X,Y) = 0$ and consequently $\Ric_{-+}(X,Y) = 0$. Clearly $\GRic(X,Y) = 0$ is not enough for $\Ric_{++}$ and $\Ric_{--}$ to vanish, as this would give us two additional conditions
\begin{align}
\Ric^{LC}(X,Y) = (H' \star H')(X,Y) &= 0, \\
H'(e^{l},\partial_{k},Y) H'(g^{-1}B(\partial_{l}),X,e^{k}) &= 0.
\end{align}
We see that $\GRic = 0$ only implies that $\Ric$ is block-diagonal with respect to the decomposition $E = V_{+} \oplus V_{-}$ induced by the generalized metric $\gm$. Finally, note that $\RS$ and $\RS_{E}$ use exactly the block diagonal components $\Ric_{++}$ and $\Ric_{--}$. One finds that
\begin{align}
\RS &= \frac{1}{2}\Ric_{++}(g^{-1}(dy^{k}),\partial_{k}) + \frac{1}{2} \Ric_{--}(g^{-1}(dy^{k}),\partial_{k}) \\ 
\RS_{E} &= \frac{1}{2} \Ric_{++}(g^{-1}(dy^{k}),\partial_{k}) - \frac{1}{2} \Ric_{--}(g^{-1}(dy^{k}),\partial_{k}).
\end{align}
\subsection{Adding the dilaton} \label{sec_dilaton}
In this section, we show how to add the scalar field $\phi \in C^{\infty}(M)$ into the picture. Also, we relate it to the procedure used in \cite{2013arXiv1304.4294G}. First note that to arbitrary $1$-form $\varphi \in \df{1}$, one can always find $K \in \df{1} \otimes \df{2}$, such that $K' = \frac{1}{6} (1 - \dim M) \cdot \varphi$. The prefactor is conventional, and shows that we have to assume that $\dim{M} > 1$. Choose
\begin{equation} \label{eq_Kdilaton}
K(X,Y,Z) := \frac{1}{6}\<\varphi,Y\> \cdot g(X,Z) - \frac{1}{6} \< \varphi, Z\> \cdot g(X,Y). 
\end{equation}
Clearly $K(X,Y,Z) = -K(X,Z,Y)$, and one can check that $K_{a} = 0$. It thus has the properties sufficient and necessary to be a part of the connection $\hcD$ as in Theorem \ref{thm_generalLC}. It is easy to see that
\begin{equation}
K'(Z) = \frac{1}{6} (1 - \dim{M}) \cdot \<\varphi,Z\>.
\end{equation}
Let $\phi \in \cif$, and choose $\varphi = 6 / (1 - \dim{M}) \cdot d\phi$. Choosing $J = 0$ and using Theorem \ref{thm_scalarcurvatures}, one obtains the scalar curvatures
\begin{align}
\RS = \RS(g) - \frac{1}{12} H'_{klm} H'^{klm} + 4 \Delta \phi - 4 \rVert d\phi \rVert^{2}_{g}, \; \RS_{E} = 0,
\end{align}
where $\Delta \phi = \Div(d \phi)$ is the Laplace-Bertrami operator corresponding to $g$. We have chosen the prefactor $1/6$ in order to make the relation to \cite{2013arXiv1304.4294G} as simple as possible. Now we recall the construction presented there. Let $\hcD^{0}$ be the Levi-Civita connection with respect to $\G_{E}$ and $[\cdot,\cdot]_{D}^{H+dB}$, where $J = K = 0$. To each $e \in \Gamma(E)$, one can assign the skew-symmetric operator $\chi^{\varphi}_{e} \in o(E)$ called the \emph{Weyl term} as 
\begin{equation}
\chi_{e}^{\varphi}e' = \< \rho^{\ast}(\varphi), e'\>_{E} \cdot e - \<e,e'\>_{E} \cdot \rho^{\ast}(\varphi),
\end{equation}
where $\rho^{\ast}: T^{\ast}M \rightarrow E$ is the inclusion of $1$-forms into $E$ defined by the short exact sequence (\ref{eq_courantSES}). Let $P_{\pm}^{0}$ be the two projectors onto the $\pm 1 $ eigenbundles $V^{0}_{\pm}$ of the generalized metric $\G_{E}$, and define 	the maps $\chi^{\pm \pm \pm}_{e} := P_{\pm}^{0}( \chi^{\varphi}_{P^{0}_{\pm}(e)} P^{0}_{\pm}(e'))$. Then, one can show that
\begin{equation}
\hcD^{\varphi}_{e}e' = \hcD^{0}_{e}e' + \frac{1}{3} \chi_{e}^{+++}e' + \frac{1}{3}  \chi_{e}^{---}e' + \chi_{e}^{+-+}e' + \chi_{e}^{-+-}e'.
\end{equation}
is exactly of the form as in the above Theorem \ref{thm_generalLC} with $J = 0$ and $K$ given by (\ref{eq_Kdilaton}).
\subsection{Background-independent gauge} \label{sec_BIG}
Consider now a very special case of a generalized metric $\gm$, where $B$ is a non-degenerate $2$-form on $M$, or equivalently a bijection $B \in \Hom(TM,T^{\ast}M)$. Let $\theta \in \Hom(T^{\ast}M,TM)$ be the corresponding inverse, $\theta = B^{-1}$. It is a well-known fact that $\theta$ is a $dB$-twisted Poisson bivector, see \cite{Severa:2001qm}. Recall, the Schouten-Nijenhuis bracket $[\theta,\theta]_{S}$ of any bivector $\theta$ with itself can be written as 
\begin{equation} \label{eq_SNbracket}
\frac{1}{2}[\theta,\theta]_{S}(\xi,\eta,\cdot) = [\theta(\xi),\theta(\eta)] - \theta( \Li{\theta(\xi)}\eta - \io_{\theta(\eta)}d\xi),
\end{equation}
for all $\xi,\eta \in \df{1}$. Also, recall that with our conventions we denote by the same character $\theta$ the map and the corresponding bivector, $\theta(\xi) = \theta(\cdot,\xi)$. Plugging $\xi = B(X)$ and $\eta = B(Z)$, the right-hand side becomes $-dB(X,Y,\cdot)$, and we thus have
\begin{equation} \label{eq_twistedPoisson}
\frac{1}{2}[\theta,\theta]_{S}(\xi,\eta,\zeta) = -dB(\theta(\xi),\theta(\eta),\theta(\zeta)),
\end{equation}
which is exactly a definition of the \emph{$dB$-twisted Poisson manifold}. It can be equivalently encoded into the Lie algebroid structure on $T^{\ast}M$ as follows. Define the (twisted Koszul) bracket $[\cdot,\cdot]_{\theta}^{dB}$ on $\df{1}$ as 
\begin{equation}
[\xi,\eta]_{\theta}^{dB} := \Li{\theta(\xi)}\eta - \io_{\theta(\eta)}d\xi + dB(\theta(\xi),\theta(\eta),\cdot).
\end{equation}
It is a straightforward calculation using (\ref{eq_SNbracket}) and (\ref{eq_twistedPoisson}) to show that $(T^{\ast}M, \theta, [\cdot,\cdot]_{\theta}^{dB})$ forms a Lie algebroid structure. In particular, note that $\theta$ has to satisfy
\begin{equation}
\theta( [\xi,\eta]_{\theta}^{dB} ) = [\theta(\xi), \theta(\eta)].
\end{equation}
This structure appeared independently in several papers on symplectic and Poisson geometry in 1980's, and was set into the framework of Courant algebroids in \cite{Severa:2001qm}. Now, recall the observation of \cite{Jurco:2013upa}, that the Seiberg-Witten open-closed relations of \cite{Seiberg:1999vs} can be interpreted as an orthogonal transformation of the generalized metric. Let $\theta \in \vf{2}$. Let $(G,\Phi)$ be the pair of fields parametrizing the generalized metric $\gm_{\theta} = (e^{\theta})^{T} \gm e^{\theta}$, where $e^{\theta}(X+\xi) = X + \theta(\xi) + \xi$. Then $(g,B)$ and $(G,\Phi)$ are related by \emph{open-closed relations}
\begin{equation}
(g+B)^{-1} = (G+\Phi)^{-1} + \theta.
\end{equation}
For $\theta = B^{-1}$, and given $(g,B)$, there is a unique solution
\begin{equation}
G = -Bg^{-1}B, \; \Phi = -B,
\end{equation}
called the \emph{background-independent gauge}.\footnote{See \cite{Seiberg:1999vs} for the origin of this name.}  We will now show that Levi-Civita connections with respect to $\gm$ and $[\cdot,\cdot]_{D}^{H}$ can be related to Levi-Civita connections for $\gm_{\theta}$ on the $\theta$-twisted Courant algebroid, which will turn out to be equipped with the twisted Dorfman bracket corresponding to the Lie algebroid bracket $[\cdot,\cdot]_{\theta}^{dB}$.

We have $\gm = (e^{-B})^{T} \G_{E} e^{-B}$, $\gm_{\theta} = (e^{B})^{T} \G^{\theta}_{E} e^{B}$, where $\G^{\theta}_{E} = \BlockDiag(G,G^{-1})$. Hence $\G^{\theta}_{E} = \sigma^{T} \G_{E} \sigma$, where $\sigma = e^{-B} e^{\theta} e^{-B}$. The map $\sigma$ has the block form 
\begin{equation} \label{eq_sigmablock}
\sigma = \bm{0}{\theta}{-B}{0}.
\end{equation}
Before we can proceed to the discussion of connections, it is clear that we have to calculate the twist of the bracket $[\cdot,\cdot]_{D}^{H'}$ by the map $\sigma$. Recall that $H' = H + dB$. Define the new Courant algebroid bracket $[\cdot,\cdot]_{D}^{\sigma}$ as
\begin{equation}
[e,e']_{D}^{\sigma} := \sigma^{-1}([\sigma(e),\sigma(e')]_{D}^{H'}).
\end{equation}
Note that the inverse map is simply $\sigma^{-1} = -\sigma$. We find, suggestively writing the sections of $E$ in the opposite order, that
\begin{equation}
[\xi + X, \eta + Y]_{D}^{\sigma} = [\xi,\eta]_{\theta}^{dB} + \Li{\xi}^{\theta}Y - \io_{\eta} d_{\theta}X - H'(\theta(\xi),\theta(\eta),\theta(\cdot))
\end{equation}
where $\Li{}^{\theta}$ and $d_{\theta}$ are the Lie derivative and exterior differentials on $\vf{\bullet}$ induced by the Lie algebroid $[\cdot,\cdot]_{\theta}^{dB}$. Define the three-vector $H'_{\theta} \in \vf{3}$ as $H'_{\theta}(\xi,\eta,\zeta) := H'(\theta(\xi),\theta(\eta),\theta(\zeta)$. Hence, we have proved that the twisted bracket $[\cdot,\cdot]_{D}^{\sigma}$ is the $H'_{\theta}$-twisted Dorfman bracket on $E = T^{\ast}M \oplus TM$ corresponding to the Lie algebroid $(T^{\ast}M, \theta, [\cdot,\cdot]_{\theta}^{dB})$. Also note that $H'_{\theta}$ can be written as
\begin{equation}
H'_{\theta} = - \frac{1}{2}[\theta,\theta]_{S} + H_{\theta},
\end{equation}
where $H_{\theta}$ is defined using the original $3$-form $H \in \df{3}$ similarly to $H'_{\theta}$. Note that the anchor map $\rho^{\sigma}$ corresponding to $[\cdot,\cdot]_{D}^{\sigma}$ is $\rho^{\sigma}(\xi + X) = \theta(\xi)$. It follows that $3$-vector $H'_{\theta}$ is automatically $d_{\theta}$-closed, which can also be shown directly.

Let us now focus on connections. Let $\cD$ be a Levi-Civita connection with respect to the generalized metric $\gm$, and $\hcD$ be the corresponding untwisted Levi-Civita connection with respect to the block diagonal generalized metric $\G_{E}$. Define the new connection $\hcD^{\theta}$ as the $\sigma$-twist:
\begin{equation}
\hcD_{e}e' := \sigma( \hcD^{\theta}_{\sigma^{-1}(e)} \sigma^{-1}(e')).
\end{equation}
It follows from the definition of the map $\sigma$ that $\hcD^{\theta}$ is the Courant algebroid connection on $E$ compatible with the metric $\G^{\theta}_{E}$, and torsion-free with respect to the bracket $[\cdot,\cdot]^{\sigma}_{D}$.

Now note that such connections can be classified in the completely the same manner as the Levi-Civita connection for ordinary twisted Dorfman bracket. We can thus employ Theorem \ref{thm_generalLC} to find the most general form of the connection $\hcD^{\theta}$, observing that it is parametrized by two tensor fields $K_{\theta} \in \vf{1} \otimes \vf{2}$ and $J_{\theta} \in \df{1} \otimes \df{2}$, such that $(K_{\theta})_{a} = (J_{\theta})_{a} = 0$. Note the exchanged role of the tangent and the cotangent bundle. It is straightforward to show that the two parametrizations are related as
\begin{equation}
K_{\theta}(\xi,\eta,\zeta) = K(\theta(\xi),\theta(\eta),\theta(\zeta)), \;
J_{\theta}(\theta(\xi),\theta(\eta),\theta(\zeta)) = J(\xi,\eta,\zeta). 
\end{equation}
Now, let us focus on the curvature operator. Define the curvature operator $\widehat{R}^{\theta}$ of $\hcD^{\theta}$ as the $\sigma$-twist of the curvature operator $\widehat{R}$ for $\hcD$: 
\begin{equation} \label{def_Rhattheta}
\widehat{R}(e,e')e'' := \sigma( \widehat{R}^{\theta}(\sigma^{-1}(e),\sigma^{-1}(e')) \sigma^{-1}(e'')).
\end{equation}
Recall that $\hcD$ has the form (\ref{eq_hRiemann}). It follows that $\widehat{R}^{\theta}$ can be written as
\begin{equation}
\widehat{R}^{\theta}(e,e')e'' = \hcD^{\theta}_{e} \hcD^{\theta}_{e'} e'' - \hcD^{\theta}_{e'} \hcD^{\theta}_{e} e'' - \hcD^{\theta}_{[e,e']_{D}^{\sigma}} e'' + \hcD^{\theta}_{\hfK_{\theta}(e,e')}e'',
\end{equation}
where $\hfK_{\theta}$ has the form $\hfK_{\theta}(e,e') = \< \hcD^{\theta}_{e_{\lambda}}e, e'\>_{E} \cdot ( \sigma^{-1}e^{-B} \circ pr_{2} \circ e^{B} \sigma)(g_{E}^{-1}(e^{\lambda}))$. But this can be rewritten, using the block form (\ref{eq_sigmablock}),  as 
\begin{equation}
\hfK_{\theta}(e,e') = \< \hcD_{e_{\lambda}}^{\theta}e,e'\>_{E} \cdot (e^{\theta} \circ pr_{1} \circ e^{-\theta})(g_{E}^{-1}(e^{\lambda}).
\end{equation}
But this means that the formula for $\widehat{R}^{\theta}$ is completely the same as for the ordinary twisted Dorfman bracket, just with the role of vector fields and one-forms interchanged. We can now define the corresponding Ricci tensor and scalar curvatures $\RS^{\theta}$ and $\RS_{E}^{\theta}$ using the generalized metric $\G^{\theta}_{E}$ and the Courant metric $g_{E}$. They can  now be calculated by Theorem \ref{thm_scalarcurvatures} using $J_{\theta}$ and $K_{\theta}$ . One obtains
\begin{align}
\label{eq_RStheta} \RS^{\theta} & = \RS^{\theta}(G^{-1}) - \frac{1}{12} (H'_{\theta})^{klm} (H'_{\theta})_{klm} + 4 \Div_{\theta}(K') - 4 \rVert K'_{\theta} \rVert^{2}_{G} - 4 \rVert J'_{\theta} \rVert_{G}^{2}, \\
\label{eq_RSEtheta} \RS^{\theta}_{E} &= -4 \Div_{\theta}( J'_{\theta}) + 8 \<J'_{\theta}, K'_{\theta}\>,
\end{align}
where $\RS^{\theta}(G^{-1})$ is the scalar curvature of the Levi-Civita connection with respect to the fiber-wise metric $G^{-1}$ on $T^{\ast}M$, torsion-free with respect to the Lie algebroid $[\cdot,\cdot]_{\theta}^{dB}$, and $\Div_{\theta}$ is the divergence using this Levi-Civita connection. We define $J'_{\theta} \in \df{1}$ and $K' \in \vf{1}$ as the partial traces using $G$, that is
\begin{equation}
J'_{\theta}(X) := J_{\theta}( G^{-1}(dy^{k}), \partial_{k}, X), \; K'_{\theta}(\xi) := K_{\theta}( G(\partial_{k}), dy^{k}, \xi).
\end{equation}
The main goal of the entire construction was to define the connection $\hcD^{\theta}$ so that its scalar curvature coincides with the scalar curvature of the original connection $\cD$. This is indeed true and straightforward to see this from (\ref{def_Rhattheta}). We conclude that
\begin{equation} \label{eq_RSequalsRStheta}
\RS = \RS^{\theta}, \; \RS_{E} = \RS^{\theta}_{E}.
\end{equation}
Let us finish this section with an example of such equality:
\begin{example}
Take the trivial twist $H = 0$, and set the tensors parametrizing the connection $\hcD$ to $J = 0$, and $K$ defining the dilaton as in (\ref{eq_Kdilaton}). Then $H'_{\theta} = -\frac{1}{2}[\theta,\theta]_{S} \equiv \Theta$, an example of the $R$-flux in the string theory. The tensor $K_{\theta}$ is now
\begin{equation}
K_{\theta}(\xi,\eta,\zeta) = \frac{1}{6}\<\varphi,\theta(\eta)\> \cdot G^{-1}(\xi,\zeta) - \frac{1}{6} \<\varphi, \theta(\zeta)\> \cdot G^{-1}(\xi,\eta).
\end{equation}
Hence the partial trace $K'_{\theta}$ is $K'_{\theta} = - \frac{1}{6}(1 - \dim{M}) \cdot \theta(\varphi)$, and for $\varphi = 6 / (1 - \dim{M}) d\phi$, we get $K'_{\theta} = -\theta(d\phi)= d_{\theta}\phi$. Plugging this into (\ref{eq_RStheta}) and using the relation (\ref{eq_RSequalsRStheta}) gives the final equality
\begin{equation}
\RS(g) - \frac{1}{12} dB_{ijk} dB^{ijk} + 4 \Delta(\phi) - 4 \rVert d\phi \rVert^{2}_{g} = \RS^{\theta}(G^{-1}) - \frac{1}{12} \Theta^{ijk} \Theta_{ijk} + 4 \Delta_{\theta}(\phi) - 4 \rVert d_{\theta}\phi \rVert^{2}_{G}. 
\end{equation}
Here $\Delta_{\theta}\phi = \Div_{\theta}( d_{\theta}\phi)$ is the Laplace-Bertrami operator corresponding to the Levi-Civita connection of the fiber-wise metric $G^{-1}$ and the Lie algebroid $[\cdot,\cdot]_{\theta}^{dB}$ on $T^{\ast}M$. This is precisely the relation of the low energy effective actions of the closed string  and the "non-geometric" string theory as it was derived on the Lie algebroid level in \cite{Blumenhagen:2012nt, Blumenhagen:2013aia}. 
\end{example}
\section{Heterotic theory} \label{sec_heterotic}
\subsection{Heterotic Courant algebroids and reduction} \label{sec_heteroticCourant}
Let us now consider a more general class of Courant algebroids. The concept of a heterotic Courant algebroid was introduced in \cite{Baraglia:2013wua} as a subclass of transitive Courant algebroids. It was shown that such a Courant algebroid can always be obtained by a reduction from an exact Courant algebroid over a principal $G$-bundle with vanishing first Pontryagin class. This reduction is a special example of Courant algebroid reductions described in \cite{Bursztyn2007726}, \cite{2015LMaPh.tmp...53S}. Finally, recall that the first Pontryagin class and its relation to Courant algebroids also appeared in \cite{2005math......9563B}.

We start with recalling some definitions. Let $(L,l,[\cdot,\cdot]_{L})$ be a transitive Lie algebroid over $M$, that is $l: L \rightarrow TM$ be surjective. The kernel $K := \ker{l}$ is naturally endowed with a totally intransitive Lie algebroid structure. We denote the corresponding bracket as $[\cdot,\cdot]_{K}$ (in fact $K$ is always a Lie algebra bundle, see \cite{Mackenzie}). Let $\<\cdot,\cdot\>_{K}$ be a non-degenerate symmetric bilinear fiber-wise form on $K$. One says that $(L,l,[\cdot,\cdot]_{L},\<\cdot,\cdot\>_{K})$ is a \emph{quadratic Lie algebroid}, if for all $k,k' \in \Gamma(K)$ and $\psi \in \Gamma(L)$
\begin{equation}
l(\psi).\<k,k'\>_{K} = \<[\psi,k]_{L},k'\>_{K} + \<k,[\psi,k']_{L}\>_{K}
\end{equation}
holds.
Note that this equation makes sense as $K$ is an ideal in $L$. 

Let $\pi: P \rightarrow M$ be a principal $G$-bundle, with $G$ being a semi-simple Lie algebra. Let $L$ be the Atiyah algebroid of $P$. The kernel $K$ of its anchor can naturally be  identified with the adjoint bundle $\g_{P} = P \times_{Ad} \g$, it comes equipped with the fiber-wise metric induced by the Killing form $c = \<\cdot,\cdot\>_{\g}$ of $\g$ (and denoted by the same symbol). Then $(L,l,[\cdot,\cdot]_{L},\<\cdot,\cdot\>_{\g})$ is a quadratic Lie algebroid. 

A Courant algebroid $(E,\rho,\<\cdot,\cdot\>_{E},[\cdot,\cdot]_{E})$ is called \emph{transitive}, if its anchor $\rho \in \Hom(E,TM)$ is surjective. In this case $E / \rho^{\ast}(T^{\ast}M)$ is naturally endowed with a quadratic Lie algebroid structure $L$. One says that $E$ is the \emph{heterotic Courant algebroid}, if $L$ is isomorphic to the Atiyah  algebroid of some principal $G$-bundle $P$ over $M$.

Similarly to an exact Courant algebroids, every heterotic Courant algebroid is isomorphic to the "model example" which is defined as follows. Let $\pi: P \rightarrow M$ be a principal $G$-bundle, where $G$ is semi-simple with the non-degenerate Killing form $c = \<\cdot,\cdot\>_{\g}$. Define the vector bundle $E' := TM \oplus \g_{P} \oplus T^{\ast}P$, $\rho' \in \End(E',TM)$ as projection onto the first factor, and the pairing $\<\cdot,\cdot\>_{E'}$ as 
\begin{equation}
\<(X,\Phi,\xi), (Y,\Phi',\eta)\>_{E'} := \eta(X) + \xi(Y) + \<\Phi,\Phi'\>_{\g},
\end{equation}
where $X,Y \in \vf{}$, $\xi,\eta \in \df{1}$ and $\Phi,\Phi' \in \Gamma(\g_{P})$. We identify the sections of $\Gamma(\g_{P})$ with $G$-equivariant functions from $P$ to $\g$. Let $A \in \Omega^{1}(P,\g)$ be a connection on $P$, and let $F \in \Omega^{2}(M,\g_{P})$ its curvature. As always, we assume the convention $F(X) := F(\cdot,X)$. Let $H_{0} \in \df{3}$. Finally, let $\cD$ denote the vector bundle connection on $\Gamma(\g_{P})$ induced by $A$, that is $\cD_{X}\Phi = X^{h}.\Phi$, $X^{h}$ being the horizontal lift of the vector field $X$ on the base manifold.
The bracket $[\cdot,\cdot]_{E'}$ is defined as
\begin{align} \label{def_bracketE'}
[(X,\Phi,\xi),(Y,\Phi',\eta)]_{E'} := \big( & [X,Y], \cD_{X}\Phi' - \cD_{Y}\Phi - F(X,Y) - [\Phi,\Phi']_{\g}, \nonumber \\
& \Li{X}\eta - \io_{Y}d\xi - H_{0}(X,Y,\cdot) \\
& - \<F(X),\Phi'\>_{\g} + \<F(Y),\Phi\>_{\g} + \<\cD{\Phi},\Phi'\>_{\g} \big). \nonumber
\end{align}
This bracket satisfies the Leibniz identity (\ref{def_courant2}) if and only if 
\begin{equation} \label{eq_heteroticmain}
dH_{0} + \frac{1}{2} \<F \^ F\>_{g} = 0,
\end{equation}
i.e. in particular, the first Pontryagin class of $P$ has to vanish. It is easy to see that $(E',\rho',\<\cdot,\cdot\>_{E'},[\cdot,\cdot]_{E'})$ defines a heterotic Courant algebroid, as the first two components of the bracket $[\cdot,\cdot]_{E'}$ is nothing but the Atiyah algebroid, where the connection $A$ was used to split the Atiyah sequence of $P$. 

The bracket described above can always be  obtained by a reduction of an exact Courant algebroid on $E = TP \oplus T^{\ast}P$, as we will now recall. Assume that $E$ is equipped with the $H$-twisted Dorfman bracket $[\cdot,\cdot]_{D}^{H}$ defined as in (\ref{def_Hdorfman}), where $H$ will be specified below. For details, see \cite{Baraglia:2013wua}. Let $\#: \g \rightarrow \vfP{}$ denote the infinitesimal version of the action of $G$ on $P$. A \emph{trivially extended action} of $\g$ on $E$ is a bracket morphism $R: \g \rightarrow \Gamma(E)$, such that the diagram
\begin{equation}
\begin{tikzcd}
\g \arrow[r,"R"] \arrow[rd, "\#"] & \Gamma(E) \arrow[d,"\rho"] \\
   & \vfP{}
\end{tikzcd}
\end{equation}
commutes, and the induced action $x \blacktriangleright e = [R(x),e]_{D}^{H}$ integrates to the $G$-action on $E$. 

Consider the following setup. Take $H \in \dfP{3}$ as 
\begin{equation} \label{eq_choiceH}
H := \pi^{\ast}(H_{0}) + \frac{1}{2} CS_{3}(A),
\end{equation}
where $CS_{3}(A) = \< \F \^ A\>_{\g} - \frac{1}{3!} \<A \^ [A \^ A]_{\g} \>_{\g}$ is the Chern-Simons $3$-form corresponding to $A$, and $\F = dA + \frac{1}{2}[A \^ A]_{\g}$. It is obviously defined so that $dH = 0$ if and only if (\ref{eq_heteroticmain}) holds. Define the map $R$ as
\begin{equation}
R(x) := \#{x} - \xi(x),
\end{equation}
where $\xi(x) = \frac{1}{2} \<A,x\>_{\g}$ is a $G$-equivariant map from $\g$ to $\dfP{1}$. Note that $\xi$ and $H$ satisfy the condition
\begin{equation}
d(\xi(x)) - \io_{\#{x}}H = 0. 
\end{equation}
The map $\xi$ is intentionally defined so that the splitting $TP \oplus T^{\ast}P$ of $E$ is $\g$-invariant with respect to the action of $\g$ induced by $R$. Indeed, one gets
\begin{equation}
[R(x),Y+\eta]_{D}^{H} = [\#{x},Y] + \Li{\#{x}}\eta + \io_{Y}d(\xi(x)) - H(\#{x},Y,\cdot) = [\#{x},Y] + \Li{\#{x}}\eta.
\end{equation}
This shows that the induced action of $\g$ on $\Gamma(E)$ is just the canonical action of $\g$ on $TP \oplus T^{\ast}P$ which can always be integrated. One easily checks that $[R(x),R(y)]_{D}^{H} = R([x,y])$, and $R$ is thus the trivially extended action of $\g$ on $E$. Moreover, observe that 
\begin{equation} \label{eq_RRpairing}
\<R(x),R(x)\>_{\g} = - 2 \io_{\#{x}}(\xi(x)) = -\<x,x\>_{\g} = - 2 c(x),
\end{equation}
where $c(x)$ is the quadratic form corresponding to the bilinear form $c = \<\cdot,\cdot\>_{\g}$. This in particular means that $H + \xi$ can be viewed as $3$-cochain in the Cartan model of $G$-equivariant cohomology of $P$, satisfying the condition
$d_{G}(H+\xi) = -c$, where $d_{G}$ is the equivariant differential. This can be in fact achieved for any trivially extended action of a compact $G$ (in general, $c$ does not have to be non-degenerate) on the exact Courant algebroid on $P$, see \cite{Bursztyn2007726}.

The reduction procedure is now following. Define the subbundle $K$ to be generated by $R(\g)$, and $K^{\perp}$ its orthogonal complement. For any trivially extended action, these are well-defined $G$-invariant subbundles of $E$. The reduced Courant algebroid $E_{red}$ is defined as the quotient
\begin{equation} \label{eq_Ered}
E_{red} := \frac{K^{\perp} / G}{(K \cap K^{\perp}) / G},
\end{equation}
so that $\Gamma(E_{red}) = \Gamma_{G}(K^{\perp}) / \Gamma_{G}(K \cap K^{\perp})$. However, note that for our example (\ref{eq_RRpairing}) hold and thus $K \cap K^{\perp} = \{0\}$. Of course, the claim is that $E' \cong E_{red}$. But this is easy to see as follows. Let $\fPsi: \Gamma(E') \rightarrow \Gamma_{G}(K^{\perp})$ be the $\cif$-linear map defined as
\begin{equation} \label{eq_isomorphismreduced}
\fPsi(X,\Phi,\xi) := X^{h} + j(\Phi) + \pi^{\ast}(\xi) + \frac{1}{2}\<A, \Phi\>_{\g},
\end{equation}
for all $(X,\Phi,\xi) \in \Gamma(E')$, where $j: \Gamma(\g_{P}) \rightarrow \Gamma_{G}(TP)$ maps $\Phi$ to the vector field $[j(\Phi)](p) = \#{[\Phi(p)]}(p)$. It is easy to see that $\fPsi$ is a bijection. The bracket on $\Gamma_{G}(K^{\perp})$ induced from $E$ is the bracket (\ref{def_bracketE'}) on $E'$. The calculation is straightforward and can be found explicitly done in \cite{Baraglia:2013wua}. Note that we do not have to assume that $G$ is compact during this entire construction, except that for a non-compact $G$, the assumed form of $R$ is not necessarily the most general trivially extended action.

To conclude this section, note that $\Gamma_{G}(K) \cong \Gamma(\g_{P})$, where the isomorphism is given by a map $R$, defined as 
\begin{equation}
R(\Phi) := j(\Phi) - \frac{1}{2}\<A,\Phi\>_{\g},
\end{equation}
which we denote by the same letter as the extended action $R: \g \rightarrow \Gamma(E)$ defined above. Thus, we can  decompose $\Gamma_{G}(E)$ as $\Gamma_{G}(E) \cong \Gamma(E') \oplus \Gamma(\g_{P})$. The twisted Dorfman bracket $[\cdot,\cdot]_{D}^{H}$ restricted  to $\Gamma_{G}(E)$ can be now written as 
\begin{equation} \label{eq_HDorfmandecomp}
\begin{split}
[(\psi,\Phi),(\psi',\Phi')]_{D}^{H} = ( & [\psi,\psi']_{E'} - [I(\Phi),I(\Phi')]_{E'}, \\
& \pi_{2}([\psi,I(\Phi')]_{E'}) - \pi_{2}([\psi',I(\Phi)]_{E}) - 2[\Phi,\Phi']_{\g}),
\end{split}
\end{equation}
where $I: \g_{P} \rightarrow E'$ is the inclusion, and $\pi_{2}: E' \rightarrow \g_{P}$ the projection onto the second factor of $E' = TM \oplus \g_{P} \oplus T^{\ast}M$.
\subsection{Generalized metric and reduction} \label{sec_gmred}
As generalized metrics form the central point of this paper, we will now briefly discuss how a generalized metric is defined on the vector bundle $E'$, and how it can be obtained via the reduction of a generalized metric on the exact Courant algebroid $E$. For a general $G$, it can be in fact defined a little bit more generally that in \cite{Baraglia:2013wua}, but for the sake of clarity we will stick to the compact case.

The first of the equivalent formulations we have given in the Section \ref{sec_genmetric} can now be easily generalized to the heterotic Courant algebroid $E'$. Indeed, we say that $\tau' \in \End(E')$ defines a generalized metric on $E'$, if $\tau'^{2} = 1$ and the formula
\begin{equation}
\gm'(\psi,\psi') = \<\psi,\tau'(\psi')\>_{E'}
\end{equation}
defines a positive definite fiber-wise metric on $E'$. Equivalently, we can say that generalized metric is a positive definite fiber-wise metric $\gm'$ on $E'$, such that viewed as a map $\gm' \in \Hom(E',E'^{\ast})$, it is orthogonal. As $\gm'$ is positive definite, it can always be  uniquely decomposed as 
\begin{equation}
\gm' = 
\begin{pmatrix}
1 & A'^{T} & C^{T} \\
0 & 1 & B'^{T} \\
0 & 0 & 1
\end{pmatrix}
\begin{pmatrix}
g_{0} & 0 & 0 \\
0 & \h & 0 \\
0 & 0 & {h}_{0}^{-1}
\end{pmatrix}
\begin{pmatrix}
1 & 0 & 0 \\
A' & 1 & 0 \\
C & B' & 1
\end{pmatrix},
\end{equation}
with $g_{0}, h_{0}$ being metric tensors on $M$, $\h$ being a positive definite fiber-wise metric on $\g_{P}$, $C \in \Hom(TM,T^{\ast}M)$, $A' \in \Omega^{1}(M,\g_{P})$ and $B' \in \Hom(\g_{P},T^{\ast}M)$. A careful analysis of the orthogonality conditions shows that $\gm'$ is a generalized metric if and only if $h_{0} = g_{0}$, $B' = -A'^{T}c$, $C = -B_{0} - \frac{1}{2} A'^{T} c A'$ for some $B_{0} \in \df{2}$ and $\h$ is compatible with $c$, i.e. $\h c^{-1} \h = c$. Thus, every generalized metric has the block form
\begin{equation} \label{eq_G'}
\gm' = \begin{pmatrix}
1 & A'^{T} & B - \frac{1}{2}A'cA'^{T} \\
0 & 1 & -c A' \\
0 & 0 & 1 
\end{pmatrix}
\begin{pmatrix}
g_{0} & 0 & 0 \\
0 & \h & 0 \\
0 & 0 & g_{0}^{-1}
\end{pmatrix}
\begin{pmatrix}
1 & 0 & 0 \\
A' & 1 & 0 \\
-B_{0} - \frac{1}{2}A'^{T} c A' & -A'^{T}c & 1
\end{pmatrix},
\end{equation}
i.e $\gm' = (e^{-\C})^{T} \G_{E'} e^{-\C}$ with $\C \in o(E')$ being skew-symmetric with respect to the pairing $g_{E'} = \<\cdot,\cdot\>_{E'}$ and 
\begin{equation} \label{eq_C}
\G_{E'} = \begin{pmatrix}
g_{0} & 0 & 0 \\
0 & \h & 0 \\
0 & 0 & g_{0}^{-1}
\end{pmatrix}, \;
\C = \begin{pmatrix}
0 & 0 & 0 \\
-A' & 0 & 0 \\
B_{0} & A'^{T}c & 0
\end{pmatrix}
\end{equation}
For compact $G$, there is in fact only possible choice for $\h$.
\begin{lemma} \label{lem_hisminusc}
For $G$ compact, the only $\h$ compatible with $c$ is $\h = -c$. 
\end{lemma}
\begin{proof}
For compact Lie group, $-c$ is positive definite. Let us work in the fixed fiber of $\g_{P}$ which is just the Lie algebra $\g$. As both $-c$ and $\h$ are positive definite, their matrices can simultaneously be diagonalized, $-c = \diag(e^{\lambda_{1}}, \dots, e^{\lambda_{k}})$, $\h = \diag(e^{\alpha_{1}}, \dots, e^{\alpha_{k}})$. Plugging into the the compatibility condition $\h c^{-1} \h = c$ then forces $2 \alpha_{i} - \lambda_{i} = \lambda_{i}$. That is $\alpha_{i} = \lambda_{i}$ and consequently $\h = -c$.
\end{proof}

Return now to the Courant algebroid $E = TP \oplus T^{\ast}P$ and describe how $\gm'$ can be obtained by reducing a generalized metric $\gm$ on $P$. Let $\tau \in \End(E)$ be a generalized metric on $E$. We say that $\tau$ is \emph{relevant for reduction} if
\begin{equation} \label{def_relevantGM}
[R(x), \tau(e)]_{D}^{H} = \tau( [R(x), e]_{D}^{H}) \;\;\mbox{and}\;\; \tau(K^{\perp}) \subseteq K^{\perp}. 
\end{equation}
The first condition is sufficient and necessary in order for $\tau$ to preserve $G$-invariant sections. It is clear that $\tau' = \tau|_{\Gamma_{G}(K^{\perp})}$ will be (using the isomorphism (\ref{eq_isomorphismreduced})) a generalized metric on $E'$ according to the above definition. 

Let us now examine the above conditions in terms of the fields $(g,B)$  parametrizing the generalized metric $\gm$ corresponding to $\tau$ on $P$.
For, consider  sections in the special form $e = \fPsi_{\pm}(X)$ for $X \in \vfP{}$, where $\fPsi_{\pm}$ were defined in (\ref{def_fPsipm}). We obtain the conditions
\begin{equation}
\pm [R(x), \fPsi_{\pm}(X)]_{D}^{H} = \tau([R(x),\fPsi_{\pm}(X)]_{D}^{H}). 
\end{equation}
Hence, $[R(x),\fPsi_{\pm}(X)]_{D}^{H} \in V_{\pm}$. Consequently 
\begin{equation}
[R(x),\fPsi_{\pm}(X)]_{D}^{H} = \fPsi_{\pm}( \rho([R(x),\fPsi_{\pm}(X)]_{D}^{H})) = \fPsi_{\pm}( [\#{x},X]). 
\end{equation}
This gives the condition $\Li{\#{x}}( (\pm g + B)(X)) = (\pm g + B)( [\#{x},X])$. Adding and subtracting these two condition gives, no so surprisingly,
\begin{equation}
\Li{\#{x}}g = 0, \; \Li{\#{x}}B = 0.
\end{equation}
Thus, $g$ and $B$ have to be $G$-invariant $2$-tensors on $P$. Equivalently, they define maps from $\Gamma_{G}(TP)$ to $\Gamma_{G}(T^{\ast}P)$. Further, note that the connection $A$ induces the isomorphisms
\begin{equation}
\Gamma_{G}(TP) \cong \vf{} \oplus \Gamma(\g_{P}), \; \Gamma_{G}(T^{\ast}P) \cong \df{1} \oplus \Gamma(\g^{\ast}_{P}),
\end{equation}
Hence, we can decompose $g$ and $B$ with respect to these direct sums, and write them in the following block form 
\begin{equation} \label{eq_gBdecompositions}
g = \bm{1}{A'^{T}}{0}{1} \bm{g_{0}}{0}{0}{\h} \bm{1}{0}{A'}{1}, \; B = \bm{B_{0}}{-B'^{T}}{B'}{\mathfrak{b}}. 
\end{equation}
Let us now examine the second condition imposed on $\tau$. With respect to the decomposition $E = K^{\perp} \oplus K$, the map $\tau$ and the pairing $g_{E}$ have the  block form
\begin{equation} \label{eq_taugEblock}
\tau = \bm{\tau'}{\tau_{1}}{0}{\tau_{K}}, \; g_{E} = \bm{g_{E'}}{0}{0}{-c}.
\end{equation}
The condition $\tau^{2} = 1$ and the orthogonality $\tau^{T} g_{E} \tau = g_{E}$ imply $\tau_{1} = 0$. Similarly to Lemma \ref{lem_hisminusc}, it follows that for a compact $G$ necessarily $\tau_{K} = 1$. In particular, we have $\tau(R(x)) = R(x)$. This requires $R(x) \in V_{+}$ and therefore necessarily $R(x) = \fPsi_{+}(\#{x})$. It follows that
\begin{equation}
\xi(x) = -(g+B)(\#{x}). 
\end{equation}
Using the explicit form of $\xi$ and $(g,B)$ decomposed as in (\ref{eq_gBdecompositions}) leads to
\begin{equation}
B' = \h A', \; \h + \mathfrak{b} = - \frac{1}{2}c. 
\end{equation}
Hence, $\mathfrak{b} = 0$, $\h = - \frac{1}{2}c$ and $B' = -\frac{1}{2} c A' $. It follows that $g$ and $B$  have the form
\begin{equation} \label{eq_gBblocksfinal}
g = \bm{1}{A'^{T}}{0}{1} \bm{g_{0}}{0}{0}{-\frac{1}{2}c} \bm{1}{0}{A'}{1}, \; B = \bm{B_{0}}{ \frac{1}{2} A'^{T}c}{-\frac{1}{2}c A'}{0}. 
\end{equation}
There are no other conditions imposed on the fields in $g$ and $B$, as the condition $\tau(R(x)) = R(x)$  already implies $\tau(K^{\perp}) \subseteq K^{\perp}$. We have 
\[ \<\tau(e),R(x)\>_{E} = -\<e, \tau(R(x))\>_{E} = -\<e, R(x)\>_{E}, \]
and thus $e \in K^{\perp}$ if and only if $\tau(e) \in K^{\perp}$. 

The easiest way how to show that $\gm'$ in the form (\ref{eq_G'}) with $\h = -c$ is indeed obtained by the reduction from $\gm$ is to write $\gm$ as $4 \times 4$ block matrix with respect to the isomorphism $\Gamma_{G}(E) \cong \vf{} \oplus \Gamma(\g_{P}) \oplus \df{1} \oplus \Gamma(\g^{\ast}_{P})$.  Note that $\fPsi: E' \rightarrow E$ can be, using this isomorphism, written in the block form
\begin{equation}
\fPsi = \begin{pmatrix}
1 & 0 & 0 \\
0 & 1 & 0 \\
0 & 0 & 1 \\
0 & \frac{1}{2}c & 0
\end{pmatrix}
\end{equation}
It is then an easy calculation to check that in fact $\gm' = \fPsi^{T} \gm \fPsi$, which proves our assertion. 

To finish this section, note that we already know that the generalized metric $\gm$ can be written as product $\gm = (e^{-B}) \G_{E} e^{-B}$. However, it is not the most convenient form for the purpose of reduction. Instead, note that one can write $g = [e^{A'}]^{T} g' e^{A'}$, where 
\begin{equation} \label{eq_H'_{A'}}
g' = \bm{g_{0}}{0}{0}{-\frac{1}{2}c}, \; e^{A'} = \bm{1}{0}{A'}{1}. 
\end{equation}
Moreover, one has $B = [e^{-A'}]^{T} B e^{-A'}$. It follows that generalized metric $\gm$ can be  written as $\gm = (e^{\A'})^{T} (e^{-B})^{T} \G'_{E} e^{-B} e^{\A'}$, where $e^{\A'}$ is defined as $e^{\A'}(X+\xi) := e^{A'}(X) + (e^{-A'})^{T}(\xi)$ and $\G'_{E}$ has $4 \times 4$ block diagonal form
\begin{equation}
\G'_{E} = \begin{pmatrix}
g_{0} & 0 & 0 & 0 \\
0 & -\frac{1}{2}c & 0 & 0 \\
0 & 0 & g_{0}^{-1} & 0 \\
0 & 0 & 0 & -2c^{-1}
\end{pmatrix}
\end{equation}
Define a new map $e^{-\B} := e^{-B} e^{\A'}$. Using this, we find its block form to be
\begin{equation}
e^{-\B} = \begin{pmatrix}
1 & 0 & 0 & 0 \\
A' & 1 & 0 & 0 \\
-B_{0} - \frac{1}{2}A'^{T}cA & -\frac{1}{2}A'^{T}c & 1 & -A'^{T} \\
\frac{1}{2}cA' & 0 & 0 & 1
\end{pmatrix}
\end{equation}
It is a very useful observation that $e^{-\B}$ is block diagonal with respect to the decomposition $\Gamma_{G}(E) \cong \Gamma(E') \oplus \Gamma(\g_{P})$. Moreover, we claim that 
\begin{equation} \label{eq_GeetoBblocks}
\G'_{E} = \bm{\G_{E'}}{0}{0}{-c}, \; e^{-\B} = \bm{e^{-\C}}{0}{0}{1},
\end{equation}
where $\G_{E'}$ and $\C \in \End(E')$ are the maps in (\ref{eq_C}). We know that any $G$-invariant section of $E$ corresponds to a $4$-tuple $(X,\Phi,\xi,\Psi)$, where $X \in \vf{}$, $\xi \in \df{1}$, $\Phi \in \Gamma(\g_{P})$, and $\Psi \in \Gamma(\g^{\ast}_{P})$. The decomposition of this $4$-tuple with respect to $\Gamma_{G}(E) = \Gamma_{G}(K^{\perp}) \oplus \Gamma_{G}(K)$ is 
\begin{equation} \label{eq_Edecomp}
(X,\Phi,\xi,\Psi) = \fPsi( X, \frac{1}{2}\Phi + c^{-1}\Psi, \xi) + R(\frac{1}{2}\Phi - c^{-1}\Psi).
\end{equation}
Using this decomposition, it is now straightforward to prove that 
\begin{equation} e^{-\B}( \fPsi(X,\Phi,\xi)) = \fPsi(e^{-\C}(X,\Phi,\xi)), \; e^{-\B}(R(\Phi)) = R(\Phi). \end{equation}
Note that (\ref{eq_GeetoBblocks}) can be used to prove once more that $\gm$ reduces to $\gm'$ on $\Gamma_{G}(K^{\perp}) \cong \Gamma(E')$.

\subsection{Twisting of the heterotic bracket} \label{sec_hettwist}
In order to proceed to the reduction of connections, we need to discuss the twisting of the bracket $[\cdot,\cdot]_{E'}$ defined by (\ref{def_bracketE'}) by the map $e^{\C} \in \Aut(E')$ with $\C$ defined as in (\ref{eq_C}). Put 
\begin{equation}
[\psi,\psi']'_{E'} = e^{-\C}[e^{\C}(\psi), e^{\C}(\psi')]_{E'},
\end{equation}
for all $\psi,\psi' \in \Gamma(E')$. Note that $e^{\C}$ is orthogonal with respect to the pairing $\<\cdot,\cdot\>_{E'}$ and  $\rho' = \rho' \circ e^{\C}$. This shows that $e^{\C}$ is an isomorphism of Courant algebroids $(E',\rho',\<\cdot,\cdot\>_{E'},[\cdot,\cdot]_{E'})$ and $(E',\rho',\<\cdot,\cdot\>_{E'},[\cdot,\cdot]'_{E'})$. We will assume the twisted bracket to take the form
\begin{equation}
[\psi,\psi']'_{E'} = [\psi,\psi]_{E'} - d\C(\psi,\psi'),
\end{equation}
where $d\C: \Gamma(E') \times \Gamma(E') \rightarrow \Gamma(E')$ is to be determined. 

Before we write down the result, let us introduce some further notation. We start by noting that a covariant derivative $\cD$ corresponding to the connection $A$ induces a differential $d_{\cD}$ on $\Gamma(\g_{P})$-valued $1$-forms. Explicitly ,
\begin{equation}
(d_{\cD}A')(X,Y) = \cD_{X}(A'(Y)) - \cD_{Y}(A'(X)) - A'([X,Y]).
\end{equation}
Moreover, we can use the fiber-wise bracket $[\cdot,\cdot]_{\g}$ on $\Gamma(\g_{P})$ to define a "covariant differential" as
\begin{equation}
\D_{\cD}A' := d_{\cD}A' + \frac{1}{2}[A' \^ A']_{\g}
\end{equation}
and a $3$-form 
\begin{equation}
\tilde{C}_{3}(A') := \< \D_{\cD}A' \^ A'\>_{\g} - \frac{1}{3!} \< [A' \^ A']_{\g} \^ A'\>_{\g}.
\end{equation}
Also, one can define a $3$-form $d\C \in \Omega^{3}(E')$ using the Courant pairing $\<\cdot,\cdot\>_{E'}$: 
\begin{equation}
d\C(\psi,\psi',\psi'') = \< d\C(\psi,\psi'), \psi''\>_{E'}. 
\end{equation}
The complete skew-symmetry of the right-hand side follows from the Courant algebroid axioms (for both brackets) and the fact that the respective Courant algebroids share the pairing and the anchor.   The result can be summarized as follows:
\begin{align}
\label{eq_dC1} d\C(X,Y,Z) &= (dB_{0} - \frac{1}{2}\tilde{C}_{3}(A') - \<F \^ A'\>_{\g})(X,Y,Z), \\
d\C(\Phi, X,Y) &= \<\Phi, [\D_{\cD}A'](X,Y)\>_{\g}, \\
d\C(\Phi,\Phi',X) &= \<[\Phi,\Phi']_{\g}, A'(X)\>_{\g},
\end{align}
with all other components being zero. The twisted bracket $[\cdot,\cdot]'_{E'}$ can be now written as 
\begin{equation}
[\psi,\psi']'_{E'} = [\psi,\psi']_{E'} - g_{E'}^{-1} d\C(\psi,\psi',\cdot)
\end{equation}
To understand the result of the twisting better, note that $A'$ is the difference of the connection $A$ and the connection induced by the $G$-invariant metric $g$ on $P$. The resulting twisting is thus nothing else but a change of splitting in the Atiyah sequence accompanied by a relevant change of the $3$-form $H_{0}$. Define 
\begin{equation}
F' := F + \D_{\cD}A', \; \cD'_{X} = \cD_{X} + [A'(X),\cdot]_{\g}, \mbox{ and} \; H'_{0} := H_{0} - \frac{1}{2}\tilde{C}_{3}(A') - \<F \^ A'\>_{\g}.
\end{equation}
The twisted bracket $[\cdot,\cdot]'_{E'}$ then takes the form
\begin{align} \label{eq_bracket'E'}
[(X,\Phi,\xi),(Y,\Phi',\eta)]'_{E'} = \big( & [X,Y], \cD'_{X}\Phi' - \cD'_{Y}\Phi - F'(X,Y) - [\Phi,\Phi']_{\g}, \nonumber \\
& \Li{X}\eta - \io_{Y}d\xi - H'_{0}(X,Y,\cdot) - dB_{0}(X,Y,\cdot)\\
& - \<F'(X),\Phi'\>_{\g} + \<F'(Y),\Phi\>_{\g} + \<\cD'{\Phi},\Phi'\>_{\g} \big). \nonumber
\end{align}
Note that $H'_{0}$ must obey the equation
\begin{equation} \label{eq_dH'_0}
dH'_{0} + \frac{1}{2}\<F' \^ F'\>_{\g} = 0.
\end{equation}
This is in accordance with the fact that if the Pontraygin class vanishes for one connection, it has to vanish for any other connection. From equations (\ref{eq_heteroticmain} and \ref{eq_dH'_0}) we find 
\begin{equation}
\<F' \^ F'\>_{\g} = \<F \^ F\>_{\g} + d( \tilde C_{3}(A') + 2 \<F \^ A'\>_{\g}).
\end{equation}
Note that this relation in fact holds on a general principal bundle $P$.
\subsection{Heterotic Levi-Civita connections} \label{sec_hetcon}
Consider now the heterotic Courant algebroid $(E', \rho', \<\cdot,\cdot\>_{E'}, [\cdot,\cdot]_{E'})$. Let $\gm'$ be a generalized metric on $E'$, which can be written as $\gm' = (e^{-\C})^{T} \G_{E'} e^{-\C}$. One can take the approach similar to the one of Section \ref{sec_LC}, but now working with $3 \times 3$ block matrices. Let $\cD'$ be a Levi-Civita connection with respect to the generalized metric $\gm'$. Define a new connection $\hcD'$:
\begin{equation}
\cD'_{\psi}\psi' = e^{\C}( \hcD'_{e^{-\C}(\psi)} e^{-\C}(\psi')). 
\end{equation}
It follows that $\hcD'$ is a Levi-Civita connection on $E'$ with respect to the generalized metric $\G_{E'}$ and the twisted heterotic bracket $[\cdot,\cdot]'_{E'}$.

It is a straightforward derivation to show that the most general connection $\hcD'$ metric compatible with both $\G_{E'}$ and $g_{E'}$ has the following block form
\begin{align}
\label{eq_hcd'1} \hcD'_{X} &= \begin{pmatrix}
\cDM_{X} & g_{0}^{-1}c(A_{X}(\cdot), \star) & g_{0}^{-1}C_{X}(\cdot,g_{0}^{-1}(\star)) \\
A_{X}(\star) & \cD_{X}^{\g} & - A_{X}(g_{0}^{-1}(\star)) \\
C_{X}(\cdot,\star) & -c(A_{X}(\cdot),\star) & \cDM_{X}
\end{pmatrix}, \\
\label{eq_hcd'2} \hcD'_{\Phi} &= \begin{pmatrix}
N_{\Phi}(\star) & g_{0}^{-1}c(B'_{\Phi}(\cdot),\star) & g_{0}^{-1} B_{\Phi}(\cdot,g_{0}^{-1}(\star)) \\
B'_{\Phi}(\star) & Q_{\Phi}(\star) & -B'_{\Phi}(g_{0}^{-1}(\star)) \\
B_{\Phi}(\cdot,\star) & -c(B'_{\Phi}(\cdot),\star) & -N_{\Phi}^{T}(\star) 
\end{pmatrix}, \\
\label{eq_hcd'3} \hcD'_{\xi} & = \begin{pmatrix}
N'_{\xi}(\star) & g_{0}^{-1}c(B''_{\xi}(\cdot),\star) & g_{0}^{-1} V_{\xi}(\cdot,g_{0}^{-1}(\star)) \\
B''_{\xi}(\star) & Q'_{\xi}(\star) & - B''_{\xi}( g_{0}^{-1}(\star)) \\
V_{\xi}(\cdot,\star) & -c(B''_{\xi}(\cdot),\star) & -{N'}_{\xi}^{T}(\star)
\end{pmatrix},
\end{align}
where $\cDM$ has to be a connection on $M$ metric compatible with $g_{0}$, and $\cD^{\g}$ has to be a connection on the vector bundle $\g_{P}$ metric compatible with $\<\cdot,\cdot\>_{\g}$. 
Next, all bottom-left corners are induced by $2$-forms on $M$, that is $C_{X},B_{\Phi},V_{\xi} \in \df{2}$. Finally, $N_{\Phi}$ and $N'_{\xi}$ have to be skew-symmetric with respect to $g_{0}$, and $Q_{\Phi},Q'_{\xi} \in \End(\g_{P})$ skew-symmetric with respect to the Killing form $\<\cdot,\cdot\>_{\g}$. 

We will not analyze the torsion-freeness condition in detail, as it is analogous to Section \ref{sec_LC}. Let us summarize the result in the form of a theorem. Denote by $H''_{0}$ the $3$-form on $M$ defined as
\begin{equation} \label{def_H''_0}
H''_{0} = H'_{0} + dB_{0} = H_{0} - \frac{1}{2}\tilde C_{3}(A') - \<F \^ A'\>_{\g} + dB_{0}.
\end{equation}

\begin{theorem} \label{thm_heteroticLC}
Let $\hcD'$ be a Courant algebroid connection on $E'$ metric compatible with the generalized metric $\G_{E'}$, parametrized as in (\ref{eq_hcd'1} - \ref{eq_hcd'3}). Then $\hcD'$ is torsion-free with respect to the twisted heterotic bracket $[\cdot,\cdot]'_{E'}$, if and only if 
\begin{align}
\cDM_{X}Y &= \cD^{0}_{X}Y + g_{0}^{-1}K'(X,Y,\cdot), \\
V_{\xi}(X,Y) &= -K'(g_{0}^{-1}(\xi),X,Y), \\
N'_{\xi}(Y) &= \frac{1}{6} H''_{0}(g_{0}^{-1}(\xi),Y,\cdot) - J'(\xi,g_{0}(Y),\cdot), \\
C_{X}(Y,Z) &= \frac{1}{3}H''_{0}(X,Y,Z) + J'(g_{0}(X),g_{0}(Y),g_{0}(Z)), \\
Q_{\Phi}(\Phi') &= -\frac{1}{3} [\Phi,\Phi']_{\g} + c^{-1}\mathfrak{j}(\Phi,\Phi',\cdot), \\
A_{X}(Y) &= \frac{1}{2}(h_{0} + C'_{0})(X,Y), \\
B''_{\xi}(Y) & = -\frac{1}{2}( h_{0} + C'_{0} + F')(g_{0}^{-1}(\xi),Y), \\
B_{\Phi}(X,Y) & = \< (C'_{0} + F')(X,Y), \Phi\>_{\g}, \\
N_{\Phi}(X) &= g_{0}^{-1} \< (C'_{0} + \frac{1}{2}F')(X,\cdot), \Phi \>_{\g}, \\
\< \cD_{X}^{\g}\Phi, \Phi'\>_{\g} & = \< \cD'_{X}\Phi, \Phi'\>_{\g} + \< \mathfrak{c}_{0}(\Phi,\Phi'), X\>, \\
\< Q'_{\xi}(\Phi),\Phi'\>_{\g} &= - \<\mathfrak{c}_{0}(\Phi,\Phi'), g_{0}^{-1}(\xi) \>, \\
\< B'_{\Phi}(X), \Phi' \>_{\g} &= \frac{1}{2}\< (\h_{0} + \mathfrak{c}_{0})(\Phi,\Phi'), X\>,
\end{align}
where the fields on the right-hand side are subject to the following conditions:
\begin{enumerate}
\item $\cDN$ is the Levi-Civita connection corresponding to $g_{0}$.
\item $K' \in \df{1} \otimes \df{2}$, such that $K'_{a} = 0$,
\item $J' \in \vf{} \otimes \vf{2}$, such that $J'_{a} = 0$, 
\item $\mathfrak{j} \in \Omega^{1}(\g_{P}) \otimes \Omega^{2}(\g_{P})$, such that $\mathfrak{j}_{a} = 0$, 
\item $h_{0} \in \Gamma(S^{2}T^{\ast}M) \otimes \Gamma(\g^{\ast}_{P})$, and $C'_{0} \in \Omega^{2}(M) \otimes \Gamma(\g^{\ast}_{P})$, otherwise arbitrary,
\item $\h_{0} \in \Gamma(S^{2} \g^{\ast}_{P}) \otimes \vf{}$, $\mathfrak{c}_{0} \in \Omega^{2}(\g^{\ast}_{P}) \otimes \vf{}$, otherwise arbitrary.
\end{enumerate}
\end{theorem}
\begin{example} \label{ex_LCminialonE'}
Unlike in the exact case, we will not proceed with the general case. Instead, we will consider only the "minimal connection", where we put most of the parametrizing fields to zero. We will comment on our choice of relating fields $C'_{0}$ and $F'$ later. Also, we postpone the discussion of the dilaton for now.  We put
\begin{equation}
K' = J' = \mathfrak{j} = h_{0} = \h_{0} = \mathfrak{c}_{0} = 0, \; C_{0} = -F'. 
\end{equation}
Let us now write the resulting connection $\hcD'$. We get
\begin{equation} \label{eq_exhLC1}
\hcD'_{X} = \begin{pmatrix}
\cDN_{X} & \frac{1}{2} g_{0}^{-1}\<F'(X),\star\>_{\g} & -\frac{1}{3}g_{0}^{-1}H''_{0}(X,g_{0}^{-1}(\star),\cdot) \\
-\frac{1}{2}F'(X,\star) & \cD'_{X} & \frac{1}{2}F'(X,g_{0}^{-1}(\star)) \\
-\frac{1}{3} H''_{0}(X,\star,\cdot) & -\frac{1}{2}\<F'(X),\star\>_{\g} & \cDN_{X}
\end{pmatrix},
\end{equation}
\begin{equation} \label{eq_exhLC2}
\hcD'_{\Phi} = \begin{pmatrix}
\frac{1}{2}g_{0}^{-1}\<F'(\star),\Phi\>_{\g} & 0 & 0 \\
0 & -\frac{1}{3}[\Phi,\star]_{\g} & 0 \\
0 & 0 & \frac{1}{2}\<F'(g_{0}^{-1}(\star)),\Phi\>_{\g}
\end{pmatrix},
\end{equation}
\begin{equation} \label{eq_exhLC3}
\hcD'_{\xi} = \begin{pmatrix}
\frac{1}{6} g_{0}^{-1}H''_{0}(g_{0}^{-1}(\xi),\star,\cdot) & 0 & 0 \\
0 & 0 & 0 \\
0 & 0 & \frac{1}{6} H''_{0}(g_{0}^{-1}(\xi), g_{0}^{-1}(\star),\cdot)
\end{pmatrix}.
\end{equation}
\end{example}
\subsection{Heterotic curvature tensor, scalar curvature} \label{sec_hetcurv}
Let us now define a version of the curvature operator suitable for  Courant algebroid connections on heterotic Courant algebroids. Define
\begin{equation} \label{def_Riemann'}
R'(\psi,\psi')\psi'' := \cD'_{\psi} \cD'_{\psi'}\psi'' - \cD'_{\psi'}\cD'_{\psi}\psi'' - \cD'_{[\psi,\psi']_{E'}} \psi'' + \cD'_{\fK'(\psi,\psi')} \psi'',
\end{equation}
where the map $\fK'$ is defined as 
\begin{equation} \label{def_fK'}
\fK'(\psi,\psi') := \< \cD'_{\psi_{\lambda}} \psi, \psi' \>_{E'} \cdot \F_{0}(g_{E'}^{-1}(\psi^{\lambda})),
\end{equation}
with $(\psi_{\lambda})_{\lambda}$ being some local frame for $E'$, and $\F_{0}(X,\Phi,\xi) = (0, \frac{1}{2}\Phi, \xi)$ is the projection onto the last two factors of $E'$, followed by the multiplication by $\frac{1}{2}$ in the $\g_{P}$ component. As $\rho' \circ \fK' = 0$, it is easy to check that $R'$ so defined is $\cif$-linear in all inputs. Moreover, it is straightforward to see that Lemma \ref{lem_Riemann} holds also in the heterotic case. The reason for the strange $\frac{1}{2}$ in $\F_{0}$, which is tricky part of the definition of curvature, will be elucidated in Section \ref{sec_conreduction}. The Ricci tensor $\Ric'$ corresponding to $R'$ is now defined by the formula (\ref{def_Ricci}). As in the exact case, introduce the two scalar curvatures $\RS'$ and $\RS'_{E'}$ as
\begin{equation}
\RS' = \Ric'( \gm'^{-1}(\psi^{\lambda}), \psi_{\lambda}), \; \RS'_{E'} = \Ric'( g_{E'}^{-1}(\psi^{\lambda}), \psi_{\lambda}). 
\end{equation}

We will now proceed with the above scalar curvatures. We will consider the case of the "untwisted" Levi-Civita connection $\hcD'$ of the form (\ref{eq_exhLC1} - \ref{eq_exhLC3}). We will do it in a way analogous to the exact case. For, note that an analogue of formula (\ref{eq_relRhatR}) holds:
\begin{equation}
R'(\psi,\psi')\psi'' = e^{\C} [ \widehat{R}'(e^{-\C}(\psi), e^{-\C}(\psi')) e^{-\C}(\psi'') ] 
\end{equation}
with $\widehat{R}'$ is defined as 
\begin{equation}
\widehat{R}'(\psi,\psi')\psi'' = \hcD'_{\psi} \hcD'_{\psi'} \psi'' - \hcD'_{\psi'} \hcD'_{\psi} \psi'' - \hcD'_{[\psi,\psi']'_{E'}} \psi'' + \hcD'_{\hfK'(\psi,\psi')} \psi''.
\end{equation}
Note that $[\cdot,\cdot]'_{E'}$ is the twisted heterotic bracket (\ref{eq_bracket'E'}), and the corresponding map $\hfK'$ has the form
\begin{equation}
\hfK'(\psi,\psi') = \< \hcD'_{\psi_{\lambda}}\psi,\psi'\>_{E'} \cdot \F_{\C}( g_{E'}^{-1}(\psi^{\lambda})),
\end{equation}
where $\F_{\C} = e^{-\C} \circ \F_{0} \circ e^{\C}$. Let $\hRic{}'$ be the Ricci tensor defined using $\widehat{R}'$. From the definition, $\RS'$ and $\RS'_{E'}$ can now be  calculated as
\begin{equation}
\RS' = \hRic{}'( \G_{E'}^{-1}(\psi^{\lambda}), \psi_{\lambda}), \; \RS'_{E'} = \hRic{}'( g_{E'}^{-1}(\psi^{\lambda}), \psi_{\lambda}). 
\end{equation}
It is interesting that although $\C$ appears explicitly in the map $\hfK'$, it cancels out in the curvatures. It appears implicitly only  in the fields of the twisted heterotic bracket $[\cdot,\cdot]'_{E'}$. Note that this significant property is due to the correct choice of the map $\F_{0}$ in the definition of $\fK'$. For example, the naive choice $\F_{0}(X,\Phi,\xi) = (0,\Phi,\xi)$ would destroy this property. As in the exact case, we write down only the result in the form of a theorem. The calculation is straightforward, yet too long to be presented here in detail. The result can be summarized as follows:
\begin{theorem} \label{thm_curvaturehetfinal}
Let $\cD$ be the Levi-Civita connection corresponding to the generalized metric $\gm'$, such that the corresponding connection $\hcD'$ has the form (\ref{eq_exhLC1} - \ref{eq_exhLC3}). The scalar curvatures $\RS'$ and $\RS'_{E'}$ defined above have the explicit form
\begin{align}
\RS' & = \RS(g_{0}) + \frac{1}{4} \<F'_{kl}, F'^{kl}\>_{\g} - \frac{1}{12} (H''_{0})_{klm} (H''_{0})^{klm} + \frac{1}{6} \dim{\g}, \\
\RS'_{E'} &= - \frac{1}{6} \dim{\g},
\end{align}
where $\RS(g_{0})$ is the Ricci scalar of the metric $g_{0}$, $H''_{0}$ is given by (\ref{def_H''_0}), and $F' = F + \D_{\cD}A'$. The notation is the one of Section \ref{sec_hettwist}. 
\end{theorem}
There is one important remark - the constant proportional to $\dim{\g}$ comes from the fact that $\<\cdot,\cdot\>_{\g}$ is the Killing form, i.e., we can use the equality
\begin{equation} \label{eq_Killing}
\<\Phi,\Phi'\>_{\g} = \< \Phi^{a}, [\Phi,[\Phi',\Phi_{a}]_{\g}]_{\g} \>, 
\end{equation}
where $(\Phi_{a})_{a=1}^{\dim{\g}}$ is some local frame on $\g_{P}$. It is the right-hand side which appears in the curvature operator, and the trace  using $\G_{E'}$ or $g_{E'}$ then gives the respective multiples of the number $\dim{\g}$. 
\begin{rem}
In example \ref{ex_LCminialonE'}, we have chosen $C_{0} = -F'$ . However, one can  show that any other choice  $C_{0} = \lambda \cdot F'$ for general $\lambda \in \R$ leads to the same scalar curvatures as in Theorem \ref{thm_curvaturehetfinal}. In particular, this  means that the prefactor $1/4$ in front of the term quadratic in $F'$ is in fact quite rigid, it cannot be changed by choosing a different "minimal" connection. 
\end{rem}
\begin{rem} \label{rem_hetdilaton}
There is a straightforward way to introduce the dilation, completely analogous to Section \ref{sec_dilaton}. Let $\phi_{0} \in \cif$ be the scalar function on $M$, and let $\varphi_{0} = 6/(1 - \dim{M}) \cdot d\phi_{0}$. Define the tensor $K' \in \df{1} \otimes \df{2}$ as 
\begin{equation} \label{eq_dilatonK'}
K'(X,Y,Z) = \frac{1}{6}\< \varphi_{0},Y\> \cdot g_{0}(X,Z) - \frac{1}{6} \<\varphi_{0},Z\> \cdot g_{0}(X,Y).
\end{equation}
Let $\cD'$ be the heterotic Levi-Civita connection as in Example \ref{ex_LCminialonE'}, modified by choosing a non-trivial tensor $K'$ defined in Theorem \ref{thm_heteroticLC} to take the form (\ref{eq_dilatonK'}). The resulting curvatures get the same dilaton terms as in Section \ref{sec_dilaton}, that is
\begin{equation}
\RS' = \RS'|_{K'=0} + 4 \Delta_{0}\phi_{0} - 4 \rVert d \phi_{0} \rVert^{2}_{g_{0}}, \; \RS'_{E} = \RS'_{E}|_{K'=0}. 
\end{equation}
\end{rem}
\begin{rem}Similarly as in the exact case, we could now compare our Ricci tensor $\Ric$ to $\GRic$ of \cite{2013arXiv1304.4294G}. Since this analogous to Section \ref{sec_RicvsGric}, we leave this to the reader.
\end{rem}
\section{Reduction from exact to heterotic} \label{sec_reduction}
\subsection{Connections relevant for reduction} \label{sec_conrel}
We have seen in Section \ref{sec_heteroticCourant} that any heterotic Courant algebroid $E'$ is obtained by the reduction of an exact Courant algebroid $E$ over a principal bundle $P$. Also, in Section \ref{sec_gmred} we have seen that any generalized metric $\gm'$ on $E'$ is obtained by the reduction of a relevant generalized metric $\gm$ on $E$. This suggests that a similar relation can be found on the level of Courant algebroid connections. 

Let $\gm$ be a generalized metric on $E$ relevant for reduction. See Section \ref{sec_gmred} for the implications of this assumption. Let $\cD$ be a Levi-Civita connection corresponding to this generalized metric. The necessary condition for the reduction is the compatibility with the group action. We say that $\cD$ is $G$-equivariant if 
\begin{equation} \label{eq_conequivariant}
[R(x), \cD_{e}e']_{D}^{H} = \cD_{[R(x),e]_{D}^{H}} e' + \cD_{e}[R(x),e']_{D}^{H}.
\end{equation} 
This condition is $C^{\infty}(P)$-linear in $e,e'$, and it is sufficient and necessary to ensure that $\cD$ can be restricted onto $\Gamma_{G}(E)$. It can be checked directly that a Levi-Civita connection $\cD$ is $G$-equivariant if and only if the tensors $J$ and $K$ of Theorem \ref{thm_generalLC} are $G$-invariant. 

Next, we would like $\cD$ to restrict onto $\Gamma(E') \cong \Gamma_{G}(K^{\perp})$ directly - that is without additional projection of the result onto $\Gamma(E')$. After a careful analysis of this assumption, one can arrive to the following definition. One says that a $G$-equivariant Levi-Civita connection $\cD$ is \emph{relevant for reduction}, if it has, with respect to the decomposition $\Gamma_{G}(E) \cong \Gamma(E') \oplus \Gamma(\g_{P})$, the following block form
\begin{equation} \label{eq_connrelevant}
\cD_{(\psi,\Phi)}  = \bm{\cD'_{\psi}}{0}{0}{ \pi_{2}( [\psi,I(\star)]_{E'}) - \frac{2}{3}[\Phi,\star]_{\g}},
\end{equation}
where $\cD'$ is a Courant algebroid connection on $E'$. It is not the most general Levi-Civita connection which restricts to $\Gamma(E')$, but this is the one which contains no additional data except for $\cD'_{\psi}$ and the brackets $[\cdot,\cdot]_{E'}$ and $[\cdot,\cdot]_{\g}$. It is defined so that $\cD$ is torsion-free with respect to $[\cdot,\cdot]_{D}^{H}$ if and only if $\cD'$ is torsion-free with respect to the heterotic bracket $[\cdot,\cdot]_{E'}$. Moreover, it is easy to see that $\cD'$ is metric compatible with the reduced generalized metric $\gm'$, hence $\cD'$ is a Levi-Civita connection on $E'$ with respect to $\gm'$. 

Compare the factor $-\frac{2}{3}$ to the factor $-\frac{1}{3}$ in (\ref{eq_exhLC2}), this is because of the factor $2$ in the bracket (\ref{eq_HDorfmandecomp}). Not every Levi-Civita connection corresponding to a relevant generalized metric $\gm$ is relevant for reduction. Of all examples, mention the "minimal" $J = K = 0$ Levi-Civita connection. In general (for $F' \neq 0$) it is {\emph{not}} relevant for reduction, as one can verify explicitly. Also, it is nether of the form (\ref{eq_connrelevant}), nor does it preserve the subspace $\Gamma_{G}(K^{\perp}) \cong \Gamma(E)$.

\begin{example} \label{ex_extensiontorelevant}
We will now construct an explicit example of a Levi-Civita connection relevant for the reduction. In particular, we will find its fields $J$ and $K$ and calculate its scalar curvatures. 

Assume that $\gm$ is a generalized metric on $E$ relevant for reduction. Let $\cD'$ be a Levi-Civita connection on $E'$ corresponding to the reduced generalized metric $\gm'$, such that the corresponding untwisted connection $\hcD'$ has the form (\ref{eq_exhLC1} - \ref{eq_exhLC3}). Now, \emph{define} the connection $\cD$ on $E$ by the formula (\ref{eq_connrelevant}). It follows that $\cD$ is an equivariant Levi-Civita connection on $E$ with respect to the generalized metric $\gm$, reducing to $\cD'$ on $E'$. Instead of $\cD$, let us examine the connection $\ocD$ defined by
\begin{equation} 
\cD_{e}e' = e^{\B}( \ocD_{e^{-\B}(e)} e^{-\B}(e')),
\end{equation} 
for all $e,e' \in \Gamma_{G}(E)$, where $e^{-\B}$ is the map defined and discussed at the end of Section \ref{sec_gmred}. It is straightforward to show that it has the block form
\begin{equation} \label{eq_connrelevant2}
\ocD_{(\psi,\Phi)} = \bm{ \hcD'_{\psi}}{0}{0}{\pi_{2}([\psi,I(\star)]'_{E'}) - \frac{2}{3}[\Phi,\star]_{\g}}.
\end{equation}
The reason why we use this connection is that it is related to the untwisted connection $\hcD$ by the orthogonal transformation $e^{\A'}$, which is block diagonal with the original splitting $E = TP \oplus T^{\ast}P$. Note that $e^{\A'}$ is just a change of frame in $TP$, accompanied with the corresponding change of dual frame in $T^{\ast}P$. It follows that $\ocD$ has to have the form
\begin{align}
\label{eq_ocDexpl1} \ocD_{X} &= \bm{\overline{\cD}^{0}_{X} + g'^{-1}\overline{K}(X,\star,\cdot)}{-\frac{1}{3} g'^{-1}\overline{H}'(X,g'^{-1}(\star),\cdot) - \overline{J}(g'(X),\star,\cdot)}{-\frac{1}{3}\overline{H}'(X,\star,\cdot) - g'\overline{J}(g'(X),g'(\star),\cdot)}{\overline{\cD}^{0}_{X} + \overline{K}(X,g'^{-1}(\star),\cdot)}, \\
\label{eq_ocDexpl2} \ocD_{\xi} & = \bm{ \frac{1}{6} g'^{-1}\overline{H}'(g'^{-1}(\xi),\star,\cdot) - \overline{J}(\xi,g'(\star),\cdot)}{ g'^{-1}\overline{K}(g'^{-1}(\xi),g'^{-1}(\star),\cdot)}{\overline{K}(g'^{-1}(\xi),\star,\cdot)}{ \frac{1}{6} \overline{H}'(g'^{-1}(\xi),g'^{-1}(\star),\cdot) - g'\overline{J}(\xi,\star,\cdot)},
\end{align}
where $g'$ is the metric defined by $g = (e^{A'})^{T} g' e^{A'}$ and $\overline{\cD}^{0}_{X}$ is the corresponding Levi-Civita connection (with respect to the twisted Lie algebroid bracket on $TP$). Tensors with bars above them are $e^{A'}$-twists of the ones parametrizing the Levi-Civita connection $\hcD$ in Theorem \ref{thm_generalLC}, that is 
\begin{align}
\overline{H}'(X,Y,Z) & = H'(e^{-A'}(X),e^{-A'}(Y),e^{-A'}(Z)), \\
\overline{K}(X,Y,Z) &= K(e^{-A'}(X), e^{-A'}(Y), e^{-A'}(Z)), \\
\overline{J}(\xi,\eta,\zeta) &= J((e^{A'})^{T}(\xi), (e^{A'})^{T}(\eta), (e^{A'})^{T}(\zeta)).
\end{align}
Let us now calculate $\overline{K}$ and $\overline{J}$ in order to determine the curvature of the connection $\cD$. We will use the notation $(X,\Phi) \in \Gamma_{G}(TP)$ and $(\xi,\Psi) \in \Gamma_{G}(T^{\ast}P)$. The tensor $\overline{K}$ is calculated as
\begin{equation}
\overline{K}( g'^{-1}(\xi,\Psi) , (Y,\Phi'), (Z,\Phi'')) = \< \ocD_{(\xi,\Psi)}(Y,\Phi'),(Z,\Phi'')\>_{E}. 
\end{equation}
Plugging into the the block form (\ref{eq_connrelevant2}), we obtain
\begin{equation}
\begin{split}
\ocD_{(\xi,\Psi)}(Y,\Phi') & = \fPsi\{ \hcD'_{(0,c^{-1}(\Psi),\xi)}(Y,\frac{1}{2}\Phi',0)\} \\
& + R\{ \pi_{2}( [(0,c^{-1}(\Psi),\xi), I( \frac{1}{2}\Phi')]'_{E'}) - \frac{2}{3}[ -c^{-1}(\Psi), \frac{1}{2}\Phi']_{\g}\}.
\end{split}
\end{equation}
Evaluating this gives
\begin{equation}
\ocD_{(\xi,\Psi)}(Y,\Phi') = (\frac{1}{2} g_{0}^{-1}\<F'(Y),\Psi\> + \frac{1}{6}g_{0}^{-1}H''_{0}(g_{0}^{-1}(\xi),Y,\cdot), -\frac{1}{3}[c^{-1}\Psi,\Phi']_{\g},0,0).
\end{equation}
As this has no $\Gamma_{G}(T^{\ast}P)$ part, it follows that 
\begin{equation}
\overline{K}(g'^{-1}(\xi,\Psi), (Y,\Phi'),(Z,\Phi'')) = 0. 
\end{equation}
The calculation of $\overline{J}$ is more elaborate, as one has to calculate first  the twisted $3$-form $\overline{H}'$. As it is straightforward, we omit the explicit calculation here. The resulting tensor $\overline{J}$ is 
\begin{equation}
\begin{split}
\overline{J}((\xi,\Psi),(\eta,\Psi'),(\zeta,\Psi'')) &= \frac{1}{3}\<\Psi, F'(g_{0}^{-1}(\eta),g_{0}^{-1}(\zeta))\> \\
& + \frac{1}{6}\<\Psi', F'(g_{0}^{-1}(\xi), g_{0}^{-1}(\zeta))\> - \frac{1}{6}\<\Psi'', F'(g_{0}^{-1}(\xi), g_{0}^{-1}(\eta)\>. 
\end{split}
\end{equation}
Now note that according to Theorem \ref{thm_scalarcurvatures}, only the partial traces $J'$ and $K'$ contribute to the scalar curvatures. Let $\overline{J}{}'$ be the partial trace of $\overline{J}$ using the metric $g'$. It follows that 
\begin{equation}
J'(\xi) = \overline{J}{}\textsl{}'( (e^{-A'})^{T}(\xi)). 
\end{equation}
It is now easy to see that $\overline{J}{}' = 0$, and consequently also $J' = 0$. Finally, we conclude from Theorem \ref{thm_scalarcurvatures} that the connection $\cD$ defined by (\ref{eq_connrelevant}) with $\cD'$ in the form (\ref{eq_hcd'1} - \ref{eq_hcd'3}) has the scalar curvatures
\begin{equation}
\RS = \RS(g) - \frac{1}{12} H'_{ijk} {H'}^{ijk}, \; \RS_{E} = 0.
\end{equation}
\end{example}

\subsection{Reduction of curvatures} \label{sec_conreduction}
Let $\cD$ be an equivariant Courant algebroid connection on $E$ relevant for reduction. By definition, it induces a Courant algebroid connection on $E'$. As it is formed from $G$-invariant objects, the corresponding curvature operator $R$ defined by (\ref{def_Riemann}) will also be  $G$-invariant. As the both fiber-wise metrics $\gm$ and $g_{E}$ are also $G$-invariant, so are the both scalar curvatures $\RS$ and $\RS_{E}$. 
For a suitable definition of the curvature operator $R'$ for the reduced connection $\cD'$, $R$ can be related to $R'$, and consequently also the respective scalar curvatures. 

To start, we have to discuss how  $\fK$ defined by (\ref{def_fK}) decomposes with respect to the splitting $\Gamma_{G}(E) \cong \Gamma(E') \oplus \Gamma(\g_{P})$. As $g_{E}$ is block diagonal with respect to this decomposition, the only important object is the projection $pr_{2} \in \End(E)$. Let $\psi = (X,\Phi,\xi) \in \Gamma(E')$ and $\Phi' \in \Gamma(\g_{P})$. Then the $G$-invariant section of $E$ corresponding to the pair $(\psi,\Phi') \in \Gamma(E') \oplus \Gamma(\g_{P})$ is 
\[ \fPsi(\psi) + R(\Phi') = (X,\Phi + \Phi',\xi, \frac{1}{2}c(\Phi - \Phi')). \]
Applying the projection gives $pr_{2}(\fPsi(\psi) + R(\Phi')) = (0,0,\xi,\frac{1}{2}c(\Phi - \Phi'))$. But this decomposes as
\[ (0,0,\xi,\frac{1}{2}c(\Phi - \Phi')) = \fPsi(0,\frac{1}{2}(\Phi - \Phi'),\xi) + R( \frac{1}{2}(\Phi' - \Phi)). \]
The map $pr_{2}$ can thus be  written as 
\begin{equation}
pr_{2}((X,\Phi,\xi),\Phi') = ((0,\frac{1}{2}(\Phi - \Phi'),\xi), \frac{1}{2}(\Phi' - \Phi)). 
\end{equation}
In the block form, we have the matrix
\begin{equation}
pr_{2} = \bm{\F_{0}}{ -\frac{1}{2}I}{-\frac{1}{2} \pi_{2}}{\frac{1}{2}},
\end{equation}
where $\F_{0}$ is exactly the map we have used in order to define $\fK'$ in (\ref{def_fK'}), $I: \g_{P} \rightarrow E'$ is the inclusion, and $\pi_{2}: E' \rightarrow \g_{P}$ the projection. This finally explains our choice in the definition of the curvature operator $R'$ of the connection $\cD'$. Now recall that $g_{E}$ has the block form (\ref{eq_taugEblock}), that is there is $-c$ in the $\g_{P}$ block. This is in contrast with $g_{E'}$ which has $c$ in its $\g_{P}$ block. Plugging in, one finds
\begin{align}
\fK((\psi,0),(\psi',0)) &= ( \fK'(\psi,\psi'), - \frac{1}{2} \< \cD'_{\Phi_{a}}\psi,\psi'\>_{E'} \cdot c^{-1}(\Phi^{a})), \\
\fK((\psi,0),(0,\Phi')) &= (0,0), \\
\fK((0,\Phi),(0,\Phi')) &= ((0, \frac{5}{6}[\Phi,\Phi']_{\g}, - \< \cD \Phi, \Phi'\>_{\g}), - \frac{5}{6}[\Phi,\Phi']_{\g}).
\end{align}
Instead of calculating the tensor $R'$, we skip the step and write the result for the block diagonal components of Ricci tensor. One obtains
\begin{align}
\Ric((\psi,0),(\psi',0)) &= \Ric'(\psi,\psi'), \\
\Ric((0,\Phi),(0,\Phi')) &= - \frac{1}{6}\<\Phi,\Phi'\>_{\g},
\end{align}
where $\mbox{Ric'}$ is the Ricci tensor corresponding to the curvature operator $R'$ of the connection $\cD'$ defined by (\ref{def_Riemann'}). Note that, once more, we have used  that the Killing form $\<\cdot,\cdot\>_{\g}$ can be written as in (\ref{eq_Killing}). Taking the trace of the Ricci tensor, we arrive to the following theorem.

\begin{theorem} \label{thm_scalarreduction}
Let $\gm$ be a generalized metric on $E$ relevant for reduction, and $\gm'$ the induced generalized metric on $E'$. Further, let  $\cD$ be an equivariant Courant algebroid connection on $E$ relevant for the reduction, and let $\cD'$ be the induced Courant algebroid connection on the reduced Courant algebroid $E'$. 

Let $\RS$ and $\RS_{E}$ be the two scalar curvatures of $\cD$, defined using $\gm$ and $g_{E}$ and the curvature operator $R$ defined as in (\ref{eq_Killing}). Let $\RS'$ and $\RS'_{E'}$ be the two scalar curvatures of $\cD'$, defined using $\gm'$ and $g_{E'}$ and the curvature operator $R'$ defined as in (\ref{def_Riemann'}). Then there holds the relation
\begin{equation} \label{eq_relationsRS}
\RS = \RS' \circ \pi + \frac{1}{6} \dim{\g}, \; \RS_{E} = \RS'_{E'} \circ \pi + \frac{1}{6} \dim{\g}.
\end{equation}
\end{theorem}

\begin{example}
Let $\cD'$ be the Levi-Civita connection on $E'$ from Example \ref{ex_LCminialonE'}. We have calculated its scalar curvatures $(\RS',\RS'_{E'})$ in Theorem \ref{thm_curvaturehetfinal}. On the other hand, we have found the corresponding Levi-Civita connection on $E$ which reduces to $\cD'$ on $E'$ in Example \ref{ex_extensiontorelevant}, and calculated its scalar curvatures $(\RS,\RS_{E})$. 
Using the theorem \ref{thm_scalarreduction} gives the following two equations:
\begin{align}
\label{eq_connrelfin1} \RS(g) - \frac{1}{12} H'_{ijk} H'^{ijk} & = \{ \RS(g_{0}) + \frac{1}{4} \<F'_{kl}, F'^{kl}\>_{\g} - \frac{1}{12} (H''_{0})_{klm} (H''_{0})^{klm} \} \circ \pi + \frac{1}{3} \dim{\g}, \\
\label{eq_connrelfin2} 0 & = (-\frac{1}{6} \dim{\g}) \circ \pi + \frac{1}{6} \dim{\g},
\end{align}
where the latter one is in this case merely a consistency check (which however proved to be useful to catch the minus signs lost on the way). Recall that 
\begin{align}
F' &= F + \D_{\cD}A', \\
H' &= dB + \pi^{\ast}(H_{0}) + \frac{1}{2} CS_{3}(A), \\
H''_{0} &= dB_{0} + H_{0} - \frac{1}{2} \tilde{C}_{3}(A') - \<F \^ A'\>_{\g}.
\end{align}
The definitions of all the objects can be found in Section \ref{sec_hettwist}.
\end{example}
\begin{rem}
In fact, the metrics $g$ and $g_{0}$ in fact do not need to be positive definite. Only, one has to \emph{assume} that there exists a decomposition of $g$ in the form (\ref{eq_gBdecompositions}). The whole reduction process can thus be, with this assumption, generalized to include any metric manifold.
\end{rem}
\begin{rem}
Note that in the entire construction, one \emph{does not} use the full Leibniz identities of any of the involved Courant algebroids. The only important property of the bracket required for the definition of the torsion and curvature is the homomorphism property (\ref{eq_anchorhom}). Thus, we can  assume that $(E,\rho,\<\cdot,\cdot\>_{E},[\cdot,\cdot]_{E})$ is a pre-Courant algebroid as defined in \cite{2004math......7399V} and studied in detail in \cite{2012arXiv1205.5898L}. For a pre-Courant algebroid, one drops the Leibniz identity (\ref{def_courant2}), and keeps only its consequence (\ref{eq_anchorhom}). 

We thus can  consider $E$ to be an $H$-twisted Dorfman bracket, where $H$ is \emph{not closed}. Consequently $H_{0}$ does not need to satisfy the condition (\ref{eq_heteroticmain}). Hence, we can consider also principal bundles with non-vanishing first Pontryagin classes, all of the above remaining valid. 
\end{rem}
\begin{rem}
Everything can be worked out for the Levi-Civita connections with dilatons, yet it is not completely straightforward. Let us start on $E$, and impose the natural condition on the dilaton function $\phi \in C^{\infty}(P)$, namely $\phi = \pi^{\ast}(\phi_{0})$ for $\phi_{0} \in \cif$. Choose $K \equiv K_{\phi}$ to be (\ref{eq_Kdilaton}), and $J$ as in Example \ref{ex_extensiontorelevant}. Unfortunately, the resulting connection $\cD$ is not relevant for reduction. 

We thus employ the opposite approach. Consider $\cD'$ to be the heterotic Levi-Civita connection with dilaton $\phi'_{0}$ as in Remark \ref{rem_hetdilaton}. Note that at this point there is no reason for $\phi'_{0}$ and $\phi_{0}$ to be equal. Denote by $K'_{0}$ the tensor (\ref{eq_dilatonK'}) adding the dilaton. Now extend the connection $\cD'$ to the connection $\cD$ on $E$ relevant for reduction using the formula (\ref{eq_connrelevant}) as in Example \ref{ex_extensiontorelevant}. By construction, $\cD$ is a Levi-Civita connection on $E$. Repeating the procedures in Example \ref{ex_extensiontorelevant}, it is straightforward to calculate the corresponding tensors $J$ and $K$. One finds 
\begin{equation} \label{eq_extensiondilatonresult}
J = J|_{K'_{0} = 0}, \; K = \pi^{\ast}( K'_{0}), 
\end{equation}
where $\pi: P \rightarrow M$ is the principal bundle projection. Note that for no choice of $\phi'_{0}$, one has $K_{\phi} = \pi^{\ast}(K'_{0})$, which explains why the original connection with dilaton is not relevant for reduction. 

Fortunately, when talking about scalar curvatures, there is a certain freedom. In particular, instead of the tensor $K_{\phi}$ (\ref{eq_Kdilaton}), one can choose any tensor $\widehat{K}_{\phi}$, as long as $K'_{\phi} = \widehat{K}'_{\phi}$. By Theorem \ref{thm_scalarcurvatures}, this does not change the scalar curvatures $\RS$ and $\RS_{E}$. Recall our assumption $\phi = \pi^{\ast}(\phi_{0})$. We claim that one can choose $\widehat{K}_{\phi} = \pi^{\ast}(K_{0})$, where $K_{0} \in \df{1} \otimes \df{2}$ is defined as 
\begin{equation}
K_{0}(X,Y,Z) = \frac{1}{6} \<\varphi_{0},Y\> \cdot g_{0}(X,Z) - \frac{1}{6} \<\varphi_{0},X\> \cdot g_{0}(X,Y),
\end{equation}
for $\varphi_{0} = 6/(1 - \dim{M}) \cdot \phi_{0}$. We have to show that $K'_{\phi} = \widehat{K}'_{\phi}$. We know from Section \ref{sec_dilaton} that $K'_{\phi} = d\phi = \pi^{\ast}(d\phi_{0})$. On the other hand, using the connection induced splitting of $\Gamma_{G}(E)$ as in Section (\ref{sec_gmred}) and the block forms (\ref{eq_gBblocksfinal}), one has
\[
\widehat{K}'_{\phi}(Z,\Phi) = \widehat{K}_{\phi}( g^{-1}(dy^{k}),0), (\partial_{k},0), (Z,\Phi)) = K_{0}( g^{-1}(dy^{k}), \partial_{k}, Z) \circ \pi= \< d\phi_{0}, Z\> \circ \pi.
\]
This proves that we can use $\widehat{K}_{\phi}$ instead of $K_{\phi}$ to get the same scalar curvatures. Moreover, according to (\ref{eq_extensiondilatonresult}), this connection is relevant for reduction to the heterotic Levi-Civita connection, and $\phi'_{0} = \phi_{0}$. One can now generalize Theorem \ref{thm_scalarreduction} to include the scalar curvatures with dilatons. In fact, we get the direct equality of both dilaton terms:
\begin{equation}
4 \Delta \phi - 4 \rVert d\phi \rVert^{2}_{g} = \{ 4 \Delta_{0}\phi_{0} - 4 \rVert d\phi_{0} \rVert^{2}_{g_{0}} \} \circ \pi.
\end{equation}
\end{rem}
\section{Double field theory curvature tensor} \label{sec_dftRiemann}
We have added this section in response to an anonymous referee. We would like to thank him for pointing out to us a relevant and interesting paper \cite{Hohm:2012mf}, which we have missed originally. The definition of a generalized Riemann tensor given in that paper can easily be compared with our definition (\ref{def_Riemann}). Note that the definition of \cite{Hohm:2012mf} is suitable for any Courant algebroid, which is an advantage over our definition (\ref{def_Riemann}). 

Let $(E,\rho,\<\cdot,\cdot\>_{E},[\cdot,\cdot]_{E})$ be any Courant algebroid, let $\cD$ be a Courant algebroid connection, and let $R^{(0)}$ denote the ``naive" curvature operator:
\begin{equation}
R^{(0)}(e,e')e'' = \cD_{e}\cD_{e'} e'' - \cD_{e'} \cD_{e} e'' - \cD_{[e,e']_{E}}e'',
\end{equation}
for all $e,e',e'' \in \Gamma(E)$. As we already know, it is $\cif$-linear in $e'$ and $e''$ but not $\cif$-linear in $e$, as we obtain 
\begin{equation} \label{eq_R0ciflin}
R^{(0)}(fe,e')e'' = f \cdot R^{(0)}(e,e')e'' - \<e,e'\>_{E} \cD_{\D{f}}e''. 
\end{equation}
The simplest way to fix this is to define the map $\fL(e,e') = \<\cD_{e_{k}}e,e'\>_{E} \cdot g_{E}^{-1}(e^{k})$, where $(e_{\lambda})_{\lambda=1}^{k}$ is an arbitrary local frame on $E$. Note that this map is $\cif$-linear in $e'$, and $\fL(fe,e') = \<e,e'\>_{E} \cdot \D{f}$. This suggests to define 
\begin{equation}
R^{(1)}(e,e')e'' = R^{(0)}(e,e')e'' + \cD_{\fL(e,e')}e''. 
\end{equation}
However, such a modification of the ``naive" curvature operator $R^{(0)}$ destroys the $\cif$-linearity in $e''$. One gets 
\begin{equation} \label{eq_R1ciflin}
R^{(1)}(e,e')(fe'') =f R^{(1)}(e,e')e'' + (\rho(\fL(e,e')).f) e'' = f R^{(1)}(e,e')e'' + \< \cD_{\D{f}}e,e'\>_{E} \cdot e''. 
\end{equation}
In our approach as described above,\footnote{Note that this approach is applicable to any (local) Leibniz algebroid.} we have solved this by using $\fK$ instead of $\fL$, which has $\rho \circ \fK = 0$. 

The nice idea of \cite{Hohm:2012mf} is to add another correcting term, namely define the operator $\~R$ as 
\begin{equation}
\~R(e,e')e'' = R^{(1)}(e,e')e'' + \<e', R^{(0)}(e'',e_{\lambda})e \>_{E} \cdot g_{E}^{-1}(e^{\lambda}). 
\end{equation}
From (\ref{eq_R0ciflin}) and (\ref{eq_R1ciflin}), it follows that $\~R$ is $\cif$-linear in all inputs. Also, in \cite{Hohm:2012mf}, the authors use the covariant generalized Riemman tensor defined as\footnote{Observe that we use a more traditional convention for the definition of a covariant Riemann tensor.} $R_{HZ}(w',w,e,e') = \<w', \~R(e,e')w\>_{E}$, for all $e,e',w,w' \in \Gamma(E)$. Explicitly:
\begin{equation} \label{eq_RHZexplicit}
\begin{split}
R_{HZ}(w',w,e,e') = & \ \< w', \cD_{e}\cD_{e'}w - \cD_{e'}\cD_{e}w - \cD_{[e,e']_{E}}w \>_{E} \\
& + \< e', \cD_{w}\cD_{w'}e - \cD_{w'}\cD_{w}e - \cD_{[w,w']_{E}}e \>_{E} \\
& + \< \cD_{e_{k}}e, e'\>_{E} \cdot \< \cD_{g_{E}^{-1}(e^{k})}w, w'\>_{E}. 
\end{split}
\end{equation}
The generalized Riemann tensor $R_{HZ}$ satisfies the analogue of Lemma \ref{lem_Riemann}, in particular it is skew-symmetric both in $(e,e')$ and $(w',w)$. Moreover, it has an additional manifest symmetry, namely
\begin{equation}
R_{HZ}(w',w,e,e') = R_{HZ}(e',e,w,w').
\end{equation}
Combined with the skew-symmetries, it gives the interchange symmetry (even for a non-vanishing torsion operator). Moreover, one can show that $\~R$ satisfies the algebraic Bianchi identity which one can write in the form:
\begin{align}
\~R(e,e')e'' + cyclic(e,e',e'') = & \  (\cD_{e}T_{G})(e',e'',e_{\lambda}) \cdot g_{E}^{-1}(e^{\lambda}) - T(e,T(e',e'')) + cyclic(e,e',e'') \nonumber \\ 
& - (\cD_{e_{\lambda}} T_{G})(e,e',e'') \cdot g_{E}^{-1}(e^{\lambda}),
\end{align}
for all $e,e',e'' \in \Gamma(E)$. Recall that $T_{G} \in \Omega^{3}(E)$ is the Gualtieri's torsion defined in Lemma \ref{lem_torsgualt} and related to the torsion operator as $T_{G}(e,e',e'') = \<e'', T(e,e')\>_{E}$. Obviously, for $T = 0$, we get the simple relation $\~R(e,e')e'' + cyclic(e,e',e'') = 0$. 

We can now compare the generalized Riemann tensor $R_{HZ}$ and the curvature operator (\ref{def_Riemann}). 
\begin{tvrz} \label{tvrz_Riemannrel}
Let $(E,\rho,\<\cdot,\cdot\>_{E},[\cdot,\cdot]_{E})$ be the standard exact Courant algebroid with $E = TM \oplus T^{\ast}M$ equipped with the $H$-twisted Dorfman bracket $[\cdot,\cdot]_{D}^{H}$. Let $R$ be the curvature operator (\ref{def_Riemann}), and $R_{HZ}$ the generalized Riemann tensor  (\ref{eq_RHZexplicit}) defined as in \cite{Hohm:2012mf}. Then 
\begin{equation}\label{HZ1}
R_{HZ}(w',w,e,e') = \<w', R(e,e')w \>_{E} + \<e', R(w,w')e\>_{E},
\end{equation} 
for all $e,e',w,w' \in \Gamma(E)$. In particular, if one defines $\Ric_{HZ}(e,e') = R_{HZ}( g_{E}^{-1}(e^{\lambda}),e,e_{\lambda},e')$, one obtains the relation with our Ricci tensor (\ref{def_Ricci}): 
\begin{equation} \label{eq_Ricirel}
\Ric_{HZ}(e,e') = \Ric(e,e') + \Ric(e',e),
\end{equation}
for all sections $e,e' \in \Gamma(E)$. 
\end{tvrz}
\begin{proof}
For $E = TM \oplus T^{\ast}M$, one can choose the local frame $\{e_{\lambda}\}_{\lambda=1}^{2n}$ to be adapted to the canonical splitting of $E$. Then the last term in the definition (\ref{eq_RHZexplicit}) splits into two terms,  each of them correcting the tensoriality of one copy of the naive operators $R^{(0)}$ in (\ref{HZ1}), giving exactly the two copies of $R$ as defined by (\ref{def_Riemann}).
\end{proof}

This proposition makes clear that there is no explicit\footnote{In other words, $B$ appears only as $dB$.} occurence of the $B$-field in the resulting scalar curvatures of Theorem \ref{thm_scalarcurvatures}. Indeed, the main trick there was to use the twisted version $\hcD$ of the connection (\ref{eq_cdandhcd}) to calculate the scalar curvature. The explicit $B$ appeared only in the correcting map $\hfK$ (\ref{eq_hfK}). Such thing does not happen for $R_{HZ}$, as the twisted torsion operator $\widehat{R}_{HZ}$ can be defined entirely in terms of the connection $\hcD$, the twisted Dorfman bracket $[\cdot,\cdot]_{E}^{H+dB}$, and the pairing $\<\cdot,\cdot\>_{E}$. It follows, using similar arguments as in (\ref{eq_scalrstwist}), that the scalar curvature  defined using $R_{HZ}$ can be calculated from $\widehat{R}_{HZ}$ and the block diagonal generalized metric $\G_{E} = \BlockDiag(g,g^{-1})$. Hence, it contains no explicit $B$. Finally, (\ref{eq_Ricirel}) shows that the scalar curvatures calculated from $\Ric_{HZ}$ are just two times our curvatures $\RS$ and $\RS_{E}$, therefore they do not contain an explicit $B$ either. 

Finally a proposition completely similar to \ref{tvrz_Riemannrel} relates the curvature operator (\ref{def_Riemann'}) to (\ref{eq_RHZexplicit}) for heterotic Courant algebroids, and explains the rather mysterious $\frac{1}{2}$ in the definition of the map $\fK'$. 
\section*{Acknowledgement}
It is a pleasure to thank Peter Bouwknegt and Peter Schupp for discussions.
The research of B.J. was supported by grant GA\v CR P201/12/G028, he would like to thank the Max Planck Institute for Mathematics, the Erwin Schrödinger International Institute for Mathematical Physics in Vienna, and the Tohoku Forum for Creativity for hospitality. 
The research of J.V. was supported by RVO: 67985840 and under Australian Research Council's Discovery Projects funding scheme (project numbers DP110100072 and DP150100008), he would like to thank the Mathematical Sciences Institute and the Australian National University for a great hospitality and a very pleasant stay in Australia. 
 \bibliography{bib}
\end{document}